\documentclass{article}
\title{Estimating Jones Polynomials is a Complete Problem for One Clean Qubit}
\author{Peter W. Shor\footnote{Mathematics Department, Massachusetts
    Institute of Technology, Cambridge, Massachusetts
    02139. \emph{shor@math.mit.edu}} \ and Stephen P. Jordan\footnote{Center
    for Theoretical Physics, Massachusetts Institute of Technology,
    Cambridge, Massachusetts 02139. \emph{sjordan@mit.edu}}} 
\date{}

\usepackage{graphicx, amssymb, amsmath, amsthm, dsfont, bm, fullpage}

\input{Qcircuit}

\newcommand{\captionfonts}{\small}

\makeatletter  
\long\def\@makecaption#1#2{%
  \vskip\abovecaptionskip
  \sbox\@tempboxa{{\captionfonts #1: #2}}%
  \ifdim \wd\@tempboxa >\hsize
    {\captionfonts #1: #2\par}
  \else
    \hbox to\hsize{\hfil\box\@tempboxa\hfil}%
  \fi
  \vskip\belowcaptionskip}
\makeatother   

\renewenvironment{proof}{\noindent \textbf{Proof:}}{$\Box$}

\begin{document}
\bibliographystyle{plain}
\maketitle

\newcommand{\braket}[2]{\langle #1|#2\rangle}
\newcommand{\Bra}[1]{\left<#1\right|}
\newcommand{\Ket}[1]{\left|#1\right>}
\newcommand{\Braket}[2]{\left< #1 \right| #2 \right>}
\renewcommand{\th}{^\mathrm{th}}
\newcommand{\tr}{\mathrm{Tr}}
\newcommand{\mx}{\begin{array}{l} \includegraphics[width=0.2in]{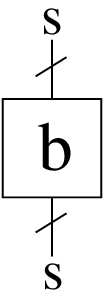} \end{array}}
\newcommand{\ttr}{\widetilde{\tr}}
\newcommand{\id}{\mathds{1}}

\newtheorem{lemma}{Lemma}
\newtheorem{theorem}{Theorem}
\newtheorem*{prop1}{Proposition 1}
\newtheorem{proposition}{Proposition}
\newtheorem{corollary}{Corollary}
\newtheorem{A}{Proposition $A_n$}
\newtheorem{B}{Proposition $B_n$}
\newtheorem{problem}{Problem}

\begin{abstract}
It is known that evaluating a certain approximation to the Jones
polynomial for the plat closure of a braid is a BQP-complete
problem. That is, this problem exactly captures the power of the
quantum circuit model\cite{Freedman, Aharonov1, Aharonov2}. The one
clean qubit model is a model of quantum computation in which all but
one qubit starts in the maximally mixed state. One clean qubit
computers are believed to be strictly weaker than standard quantum
computers, but still capable of solving some classically intractable
problems \cite{Knill}. Here we show that evaluating a certain
approximation to the Jones polynomial at a fifth root of unity for the
trace closure of a braid is a complete problem for the one clean qubit
complexity class. That is, a one clean qubit computer can approximate
these Jones polynomials in time polynomial in both the number of
strands and number of crossings, and the problem of simulating a one
clean qubit computer is reducible to approximating the Jones
polynomial of the trace closure of a braid.
\end{abstract}

\section{One Clean Qubit}
\label{DQC1}

The one clean qubit model of quantum computation originated as an
idealized model of quantum computation on highly mixed initial states,
such as appear in NMR implementations\cite{Knill, Ambainis}. In this
model, one is given an initial quantum state consisting of a
single qubit in the pure state $\ket{0}$, and $n$ qubits in the
maximally mixed state. This is described by the density matrix
\[
\rho = \ket{0}\bra{0} \otimes \frac{I}{2^n}.
\]

One can apply any polynomial-size quantum circuit to $\rho$, and
then measure the first qubit in the computational basis. Thus, if the
quantum circuit implements the unitary transformation $U$, the
probability of measuring $\ket{0}$ will be 
\begin{equation}
\label{experiment}
p_0 = \tr[(\ket{0}\bra{0}\otimes I) U \rho U^\dag] =
2^{-n} \tr[(\ket{0}\bra{0} \otimes I) U (\ket{0}\bra{0} \otimes I)  U^\dag].
\end{equation}

Computational complexity classes are typically described using
decision problems, that is, problems which admit yes/no
answers. This is mathematically convenient, and the implications for
the complexity of non-decision problems are usually straightforward to
obtain (\emph{cf.} \cite{Papadimitriou}). The one clean qubit
complexity class consists of the decision problems which can be solved
in polynomial time by a one clean qubit machine with correctness 
probability of at least $2/3$. The experiment described in 
equation  \ref{experiment} can be repeated polynomially many
times. Thus, if $p_1 \geq 1/2 + \epsilon$ for instances to which the
answer is yes, and $p_1 \leq 1/2 - \epsilon$ otherwise, then by
repeating the experiment $\mathrm{poly}(1/\epsilon)$ times and taking
the majority vote one can achieve $2/3$ probability of
correctness. Thus, as long as $\epsilon$ is at least an inverse
polynomial in the problem size, the problem is contained in the one
clean qubit complexity class. Following \cite{Knill}, we will refer to
this complexity class as DQC1.

A number of equivalent definitions of the one clean qubit complexity
class can be made. For example, changing the pure part of the
initial state and the basis in which the final measurement is
performed does not change the resulting complexity class. Less
trivially, allowing logarithmically many clean qubits results in
the same class, as discussed below. It is essential that on a given
copy of $\rho$, measurements are performed only at the end of the
computation. Otherwise, one could obtain a pure state
by measuring $\rho$ thus making all the qubits ``clean'' and
re-obtaining BQP. Remarkably, it is not necessary to have even one
fully polarized qubit to obtain the class DQC1. As shown in
\cite{Knill}, a single partially polarized qubit suffices.

In the original definition\cite{Knill} of DQC1 it is assumed that a
classical computer generates the quantum circuits to be applied to the
initial state $\rho$. By this definition DQC1 automatically contains
the complexity class P. However, it is also interesting to consider a
slightly weaker one clean qubit model, in which the classical computer
controlling the quantum circuits has only the power of NC1. The
resulting complexity class appears to have the interesting property
that it is incomparable to P. That is, it is not contained in P nor
does P contain it. We suspect that our algorithm and hardness proof
for the Jones polynomial carry over straightforwardly to this
NC1-controlled one clean qubit model. However, we have not pursued
this point.

Any $2^n \times 2^n$ unitary matrix can be decomposed as a linear
combination of $n$-fold tensor products of Pauli matrices. As
discussed in \cite{Knill}, the problem of estimating a coefficient in
the Pauli decomposition of a quantum circuit to polynomial accuracy is
a DQC1-complete problem. Estimating the normalized trace of a quantum
circuit is a special case  of this, and it is also DQC1-complete. This
point is discussed in \cite{Shepherd}. To make our presentation
self-contained, we will sketch here a proof that trace estimation is
DQC1-complete. Technically, we should consider the decision problem of
determining whether the trace is greater than a given
threshold. However, the trace estimation problem is easily reduced to
its decision version by the method of binary search, so we will
henceforth ignore this point.

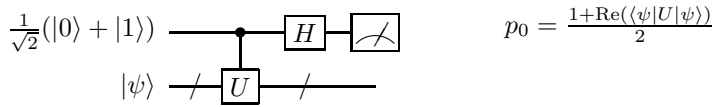
\begin{figure}
\begin{center}
\[
\begin{array}{lll}
\mbox{\Qcircuit @C=1em @R=.7em {
  \lstick{\frac{1}{\sqrt{2}}(\ket{0}+\ket{1})} & \qw & \ctrl{1} & \gate{H} & \meter \\
  \lstick{\ket{\psi}} & {/} \qw  & \gate{U} & {/} \qw & \qw & }} &
  \quad & p_0 = \frac{1+\mathrm{Re}(\bra{\psi}U\ket{\psi})}{2}
\end{array} \]
\caption{\label{Hadamard} This circuit implements the Hadamard
  test. A horizontal line represents a qubit. A horizontal line with a
  slash through it represents a register of multiple qubits. The
  probability $p_0$ of measuring $\ket{0}$ is as shown  above. Thus,
  one can obtain the real part of $\bra{\psi}U\ket{\psi}$ 
  to precision $\epsilon$ by making $O(1/\epsilon^2)$ measurements and
  counting what fraction of the measurement outcomes are
  $\ket{0}$. Similarly, if the control bit is instead initialized to 
  $\frac{1}{\sqrt{2}} (\ket{0} - i \ket{1})$, one can estimate the
  imaginary part of $\bra{\psi}U\ket{\psi}$.}
\end{center}
\end{figure}

First we'll show that trace estimation is contained in DQC1. Suppose we
are given a quantum circuit on $n$ qubits which consists of
polynomially many gates from some finite universal gate set. Given a
state $\ket{\psi}$ of $n$ qubits, there is a standard technique for
estimating $\bra{\psi} U \ket{\psi}$, called the Hadamard
test\cite{Aharonov1}, as shown in figure \ref{Hadamard}. Now suppose
that we use the circuit from figure \ref{Hadamard}, but choose
$\ket{\psi}$ uniformly at random from the $2^n$ computational basis
states. Then the probability of getting outcome $\ket{0}$ for a given
measurement will be
\[
p_0 = \frac{1}{2^n} \sum_{x \in \{0,1\}^n} \frac{1 +
  \mathrm{Re}(\bra{x} U \ket{x})}{2} = \frac{1}{2} +
\frac{\mathrm{Re}(\tr \ U)}{2^{n+1}}.
\]
Choosing $\ket{\psi}$ uniformly at random from the $2^n$ computational
basis states is exactly the same as inputting the density matrix $I/2^n$
to this register. Thus, the only clean qubit is the control
qubit. Trace estimation is therefore achieved in the one clean qubit
model by converting the given circuit for $U$ into a circuit for
controlled-$U$ and adding Hadamard gates on the control bit. One can
convert a circuit for $U$ into a circuit for controlled-$U$ by
replacing each gate $G$ with a circuit for controlled-$G$. The
overhead incurred is thus bounded by a constant factor
\cite{Nielsen}. 

Next we'll show that trace estimation is hard for DQC1. Suppose we are
given a classical description of a quantum circuit implementing some
unitary transformation $U$ on $n$ qubits. As shown in equation
\ref{experiment}, the probability of obtaining outcome $\ket{0}$ from
the one clean qubit computation of this circuit is proportional to the
trace of the non-unitary operator $(\ket{0}\bra{0} \otimes I) U
(\ket{0}\bra{0} \otimes I)  U^\dag$, which acts on $n$
qubits. Estimating this can be achieved by estimating the trace of
\[
U' = \begin{array}{l} \Qcircuit @C=1em @R=.7em {
 &  \qw   & \multigate{1}{U^\dag} & \ctrl{2} & \multigate{1}{U} & \ctrl{3} & \qw    & \qw \\
 & {/}\qw & \ghost{U^\dag}        & \qw      & \ghost{U}        & \qw      & {/}\qw & \qw \\
 &  \qw   & \qw                   & \targ    & \qw              & \qw      & \qw    & \qw \\
 &  \qw   & \qw                   & \qw      & \qw              & \targ    & \qw    & \qw
} \end{array}
\]
which is a unitary operator on $n+2$ qubits. This suffices because
\begin{equation}
\label{tracform}
\tr[(\ket{0}\bra{0} \otimes I) U (\ket{0}\bra{0} \otimes I)
  U^\dag] = \frac{1}{4} \tr[ U' ].
\end{equation}
To see this, we can think in terms of the computational basis:
\[
\tr[U'] = \sum_{x \in \{0,1\}^n} \bra{x} U' \ket{x}.
\]
If the first qubit of $\ket{x}$ is $\ket{1}$, then
the rightmost CNOT in $U'$ will flip the lowermost qubit. The resulting state
will be orthogonal to $\ket{x}$ and the corresponding  matrix element
will not contribute to the trace. Thus this CNOT gate simulates the
initial projector $\ket{0}\bra{0} \otimes I$ in equation
\ref{tracform}. Similarly, the other CNOT in $U'$ simulates the other
projector in equation \ref{tracform}.

The preceding analysis shows that, given a description of a quantum
circuit implementing a unitary transformation $U$ on $n$-qubits, the
problem of  approximating $\frac{1}{2^n} \tr \ U$ to within $\pm
\frac{1}{\mathrm{poly}(n)}$ precision is DQC1-complete.

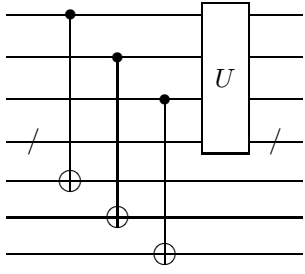
\begin{figure}
\begin{center}
\mbox{ 
\Qcircuit @C=1em @R=.7em {
  & \qw    & \ctrl{4} & \qw      & \qw      & \multigate{3}{U} & \qw    & \qw \\
  & \qw    & \qw      & \ctrl{4} & \qw      & \ghost{U}        & \qw    & \qw \\
  & \qw    & \qw      & \qw      & \ctrl{4} & \ghost{U}        & \qw    & \qw \\
  & {/}\qw & \qw      & \qw      & \qw      & \ghost{U}        & {/}\qw & \qw \\ 
  & \qw    & \targ    & \qw      & \qw      & \qw              & \qw    & \qw \\
  & \qw    & \qw      & \targ    & \qw      & \qw              & \qw    & \qw \\
  & \qw    & \qw      & \qw      & \targ    & \qw              & \qw    & \qw
} }
\caption{\label{ancillas} Here CNOT gates are used to simulate 3 clean
  ancilla qubits.}
\end{center}
\end{figure}

Some unitaries may only be efficiently implementable using ancilla
bits. That is, to implement $U$ on $n$-qubits using a quantum circuit,
it may be most efficient to construct a circuit on $n+m$ qubits which
acts as $U \otimes I$, provided that the $m$ ancilla qubits are all
initialized to $\ket{0}$. These ancilla qubits are used as work bits
in intermediate steps of the computation. To estimate the trace of
$U$, one can construct a circuit $U_a$ on $n+2m$ qubits by adding CNOT
gates controlled by the $m$ ancilla qubits and acting on $m$ extra
qubits, as shown in figure \ref{ancillas}. This simulates the presence
of $m$ clean ancilla qubits, because if any of the ancilla qubits is
in the $\ket{1}$ state then the CNOT gate will flip the corresponding
extra qubit, resulting in an orthogonal state which will not
contribute to the trace.

With one clean qubit, one can estimate the trace of $U_a$ to a
precision of $\frac{2^{n+2m}}{\mathrm{poly}(n,m)}$. By construction, 
$\tr[U_a] = 2^m \tr[U]$. Thus, if $m$ is logarithmic in $n$, then
one can obtain $\tr[U]$ to precision $\frac{2^n}{\mathrm{poly}(n)}$,
just as can be obtained for circuits not requiring ancilla
qubits. This line of reasoning also shows that the $k$-clean qubit
model gives rise to the same complexity class as the one clean qubit
model, for any constant $k$, and even for $k$ growing logarithmically
with $n$.

It seems unlikely that the trace of these exponentially large unitary
matrices can be estimated to this precision on a classical computer in
polynomial time. Thus it seems unlikely that DQC1 is contained in
P. (For more detailed analysis of this point see \cite{Datta}.)
However, it also seems unlikely that DQC1 contains all of BQP. In
other words, one clean qubit computers seem to provide exponential
speedup over classical computation for some problems despite being
strictly weaker than standard quantum computers.

\section{Jones Polynomials}

A knot is defined to be an embedding of the circle in $\mathbb{R}^3$
considered up to continuous transformation (isotopy). More
generally, a link is an embedding of one or more circles in
$\mathbb{R}^3$ up to isotopy. In an oriented knot or link, one of
the two possible traversal directions is chosen for each circle.
Some examples of knots and links are shown in figure \ref{knots}. One
of the fundamental tasks in knot theory is, given two representations
of knots, which may appear superficially different, determine whether
these both represent the same knot. In other words, determine whether
one knot can be deformed into the other without ever cutting the
strand.

\begin{figure}
\begin{center}
\includegraphics[width=0.43\textwidth]{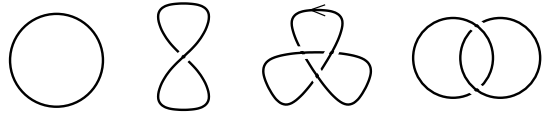}
\caption{\label{knots} Shown from left to right are the unknot,
  another representation of the unknot, an oriented trefoil knot, and
  the Hopf link. Broken lines indicate undercrossings.}
\end{center}
\end{figure} 

Reidemeister showed in 1927 that two knots are the same if and only if
one can be deformed into the other by some sequence constructed from
three elementary moves, known as the Reidemeister moves, shown in
figure \ref{Reidemeister}. This reduces the problem of distinguishing
knots to a combinatorial problem, although one for which no efficient
solution is known. In some cases, the sequence of Reidemeister moves
needed to show equivalence of two knots involves intermediate steps
that increase the number of crossings. Thus, it is very difficult to
show upper bounds on the number of moves necessary. The most
thoroughly studied knot equivalence problem is the problem of deciding
whether a given knot is equivalent to the unknot. Even showing the
decidability of this problem is highly nontrivial. This was achieved
by Haken in 1961\cite{Haken}. In 1998 it was shown by Hass, Lagarias,
and Pippenger that the problem of recognizing the unknot is
contained in NP\cite{Hass}.

\begin{figure}
\begin{center}
\includegraphics[width=0.85\textwidth]{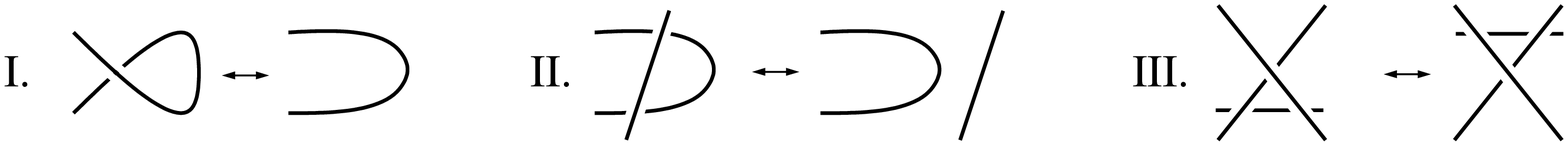}
\caption{\label{Reidemeister} Two knots are the same
  if and only if one can be deformed into the other by some sequence
  of the three Reidemeister moves shown above.}
\end{center}
\end{figure} 

A knot invariant is any function on knots which is invariant under the
Reidemeister moves. Thus, a knot invariant always takes the same
value for different representations of the same knot, such as the two
representations of the unknot shown in figure \ref{knots}. In general,
there can be distinct knots which a knot invariant fails to
distinguish. 

One of the best known knot invariants is the Jones polynomial,
discovered in 1985 by Vaughan Jones\cite{Jones}. To any oriented knot
or link, it associates a Laurent polynomial in the variable
$t^{1/2}$. The Jones polynomial has a degree in $t$ which grows at
most linearly with the number of crossings in the link. The
coefficients are all integers, but they may be exponentially
large. Exact evaluation of Jones polynomials at all but a few special
values of $t$ is \#P-hard\cite{Jaeger}. The Jones polynomial can be
defined recursively by a simple ``skein'' relation. However, for our
purposes it will be more convenient to use a definition in terms of a
representation of the braid group, as discussed below.

\begin{figure}[!htbp]
\begin{center}
\includegraphics[width=0.6\textwidth]{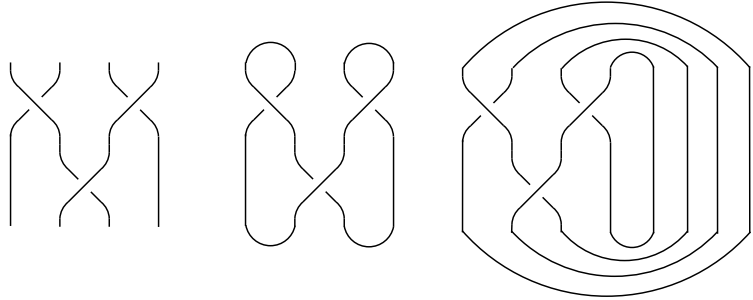}
\caption{\label{trace_plat} Shown from left to right are a braid, its
  plat closure, and its trace closure.}
\end{center}
\end{figure} 

To describe in more detail the computation of Jones polynomials we
must specify how the knot will be represented on the
computer. Although an embedding of a circle in $\mathbb{R}^3$ is a
continuous object, all the topologically relevant information about a
knot can be described in the discrete language of the braid
group. Links can be constructed from braids by joining the free
ends. Two ways of doing this are taking the plat closure and
the trace closure, as shown in figure \ref{trace_plat}. Alexander's
theorem states that any link can be constructed as the trace closure
of some braid. Any link can also be constructed as the plat closure of
some braid. This can be easily proven as a corollary to Alexander's
theorem, as shown in figure \ref{trace_to_plat}.

\begin{figure}[!htbp]
\begin{center}
\includegraphics[width=0.45\textwidth]{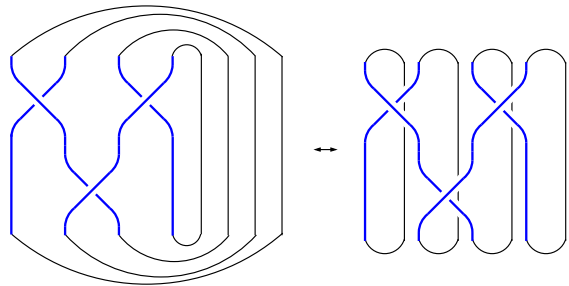}
\caption{\label{trace_to_plat} A trace closure of a braid on $n$
  strands can be converted to a plat closure of a braid on $2n$
  strands by moving the ``return'' strands into the braid.}
\end{center}
\end{figure} 

\begin{figure}[!htbp]
\begin{center}
\includegraphics[width=0.8\textwidth]{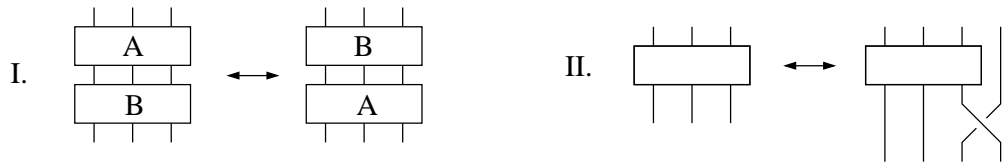}
\caption{\label{Markov} Shown are the two Markov moves. Here the
  boxes represent arbitrary braids. If a function on braids is
  invariant under these two moves, then the corresponding function on
  links induced by the trace closure is a link invariant.}
\end{center}
\end{figure} 

Given that the trace closure provides a correspondence between links
and braids, one may attempt to find functions on braids which yield
link invariants via this correspondence. Markov's theorem shows that a
function on braids will yield a knot invariant provided it is invariant
under the two Markov moves, shown in figure \ref{Markov}. Thus the
Markov moves provide an analogue for braids of the Reidemeister moves
on links. The constraints imposed by invariance under the Reidemeister
moves are enforced in the braid picture jointly by invariance under
Markov moves and by the defining relations of the braid group.

A linear function $f$ satisfying $f(AB) = f(BA)$ is called a
trace. The ordinary trace on matrices is one such function. Taking a
trace of a representation of the braid group yields a function on
braids which is invariant under Markov move I. If the trace and
representation are such that the resulting function is also invariant
under Markov move II, then a link invariant will result. The Jones
polynomial can be obtained in this way.

In \cite{Aharonov1}, Aharonov, \emph{et al.} show that an
additive approximation to the Jones polynomial of the plat or trace
closure of a braid at $t=e^{i 2 \pi/k}$ can be computed on a quantum
computer in time which scales polynomially in the number of strands
and crossings in the braid and in $k$. In \cite{Aharonov2, Wocjan}, it
is shown that for plat closures, this problem is BQP-complete. The
complexity of approximating the Jones polynomial for trace closures
was left open, other than showing that it is contained in BQP.

The results of \cite{Aharonov1, Aharonov2, Wocjan} reformulate and
generalize the previous results of Freedman \emph{et al.}
\cite{Freedman, Freedman2}, which show that certain
approximations of Jones polynomials are BQP-complete. The work of
Freedman \emph{et al.} in turn builds upon Witten's discovery of a
connection between Jones polynomials and topological quantum field
theory \cite{Witten}. Recently, Aharonov \emph{et al.} have
generalized further, obtaining an efficient quantum algorithm for
approximating the Tutte polynomial for any planar graph, at any point
in the complex plane, and also showing BQP-hardness at some points 
\cite{Aharonov3}. As special cases, the Tutte polynomial includes the
Jones polynomial, other knot invariants such as the HOMFLY polynomial,
and partition functions for some physical models such as the Potts
model.

The algorithm of Aharonov \emph{et al.} works by obtaining the Jones
polynomial as a trace of the path model representation of the braid
group. The path model representation is unitary for $t=e^{i 2\pi/k}$
and, as shown in \cite{Aharonov1}, can be efficiently implemented by
quantum circuits. For computing the trace closure of a braid the
necessary trace is similar to the ordinary matrix trace except that
only a subset of the diagonal elements of the unitary implemented by
the quantum circuit are summed, and there is an additional weighting
factor. For the plat closure of a braid the computation instead
reduces to evaluating a particular matrix element of the quantum
circuit. Aharonov \emph{et al.} also use the path model representation
in their proof of BQP-completeness.

Given a braid $b$, we know that the problem of approximating the Jones
polynomial of its plat closure is BQP-hard. By Alexander's
theorem, one can obtain a braid $b'$ whose trace closure is the same
link as the plat closure of $b$. The Jones polynomial depends only on
the link, and not on the braid it was derived from. Thus, one may ask
why this doesn't immediately imply that estimating the Jones
polynomial of the trace closure is a BQP-hard problem. The answer lies
in the degree of approximation. As discussed in section
\ref{conclusion}, the BQP-complete problem for plat closures is to
approximate the Jones polynomial to a certain precision which depends
exponentially on the number of strands in the braid. The number of
strands in $b'$ can be larger than the number of strands in $b$, hence
the degree of approximation obtained after applying Alexander's
theorem may be too poor to solve the original BQP-hard problem.

The fact that computing the Jones polynomial of the trace closure of a
braid can be reduced to estimating a generalized trace of a unitary
operator and the fact that trace estimation is DQC1-complete suggest a
connection between Jones polynomials and the one clean qubit
model. Here we find such a connection by showing that evaluating a
certain approximation to the Jones polynomial of the
trace closure of a braid at a fifth root of unity is
DQC1-complete. The main technical difficulty is obtaining the Jones 
polynomial as a trace over the entire Hilbert space rather than as a
summation of some subset of the diagonal matrix elements. To do this
we will not use the path model representation of the braid group, but
rather the Fibonacci representation, as described in the next section.

\section{Fibonacci Representation}
\label{Fibonacci}

The Fibonacci representation $\rho_F^{(n)}$ of the braid group $B_n$
is described in \cite{Kauffman} in the context of Temperley-Lieb
recoupling theory. Temperley-Lieb recoupling theory describes two
species of idealized ``particles'' denoted by $p$ 
and $*$. We will not delve into the conceptual and mathematical
underpinnings of Temperley-Lieb recoupling theory. For present
purposes, it will be sufficient to regard it as a formal procedure for
obtaining a particular unitary representation of the braid group whose
trace yields the Jones polynomial at $t = e^{i 2 \pi/5}$. Throughout
most of this paper it will be more convenient to express the Jones
polynomial in terms of $A = e^{-i 3 \pi/5}$, with $t$ defined by $t =
A^{-4}$. 

It is worth noting that the Fibonacci representation is a special case
of the path model representation used in \cite{Aharonov1}. The path
model representation applies when $t = e^{i 2 \pi/k}$ for any integer
$k$, whereas the Fibonacci representation is for $k=5$. The
relationship between these two representations is briefly discussed in
appendix \ref{rep_relations}. However, for the sake of making our
discussion self contained, we will derive all of our results directly
within the Fibonacci representation.

\begin{figure}
\begin{center}
\includegraphics[width=0.25\textwidth]{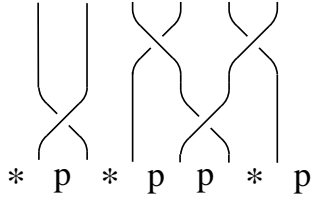}
\caption{\label{pstar} For an $n$-strand braid we can write a length
  $n+1$ string of $p$ and $*$ symbols across the base. The string may
  have no two $*$ symbols in a row, but can be otherwise arbitrary.}
\end{center}
\end{figure} 

\begin{figure}
\begin{center}
\includegraphics[width=0.5\textwidth]{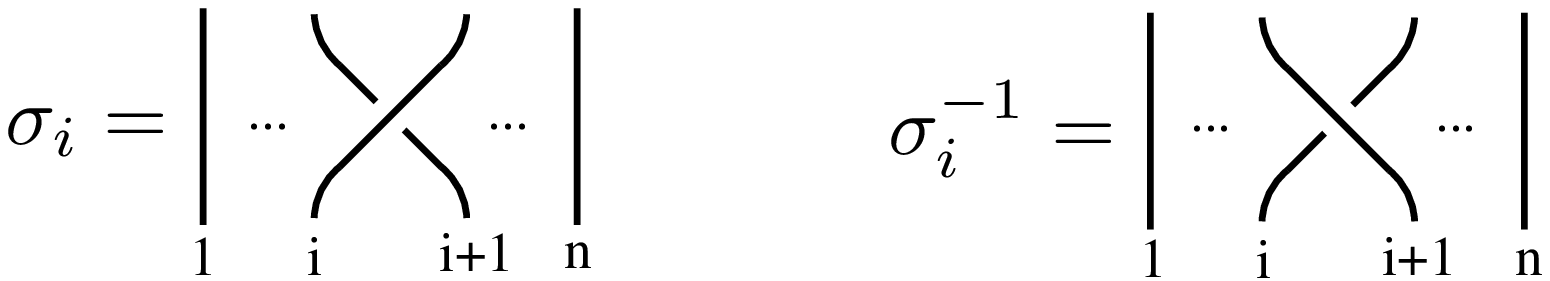}
\caption{\label{braids} $\sigma_i$ denotes the elementary crossing of
  strands $i$ and $i+1$. The braid group on $n$ strands $B_n$ is
  generated by $\sigma_1 \ldots \sigma_{n-1}$, which satisfy the
  relations $\sigma_i \sigma_j = \sigma_j \sigma$ for $|i-j|>1$ and
  $\sigma_{i+1} \sigma_i \sigma_{i+1} = \sigma_i \sigma_{i+1}
  \sigma_i$ for all $i$. The group operation corresponds to
  concatenation of braids.}
\end{center}
\end{figure}

Given an $n$-strand braid $b \in B_n$, we can write a length $n+1$
string of $p$ and $*$ symbols across the base as shown in figure
\ref{pstar}. These strings have the restriction that no two $*$
symbols can be adjacent. The number of such strings is $f_{n+3}$,
where $f_n$ is the $n\th$ Fibonacci number, defined so that $f_1 = 1$,
$f_2 = 1$, $f_3 = 2,\ldots$ Thus the formal linear combinations of
such strings form an $f_{n+3}$-dimensional vector space. For each $n$,
the Fibonacci representation $\rho_F^{(n)}$ is a homomorphism from
$B_n$ to the group of unitary linear transformations on this
space. We will describe the Fibonacci representation in terms of its
action on the elementary crossings which generate the braid group, as
shown in figure \ref{braids}.

The elementary crossings correspond to linear operations which mix
only those strings which differ by the symbol beneath the
crossing. The linear transformations have a local structure, so that
the coefficients for the symbol beneath the crossing to be changed or
unchanged depend only on that symbol and its two neighbors. For
example, using the notation of \cite{Kauffman}, 
\begin{equation}
\label{picto}
\includegraphics[width=0.35\textwidth]{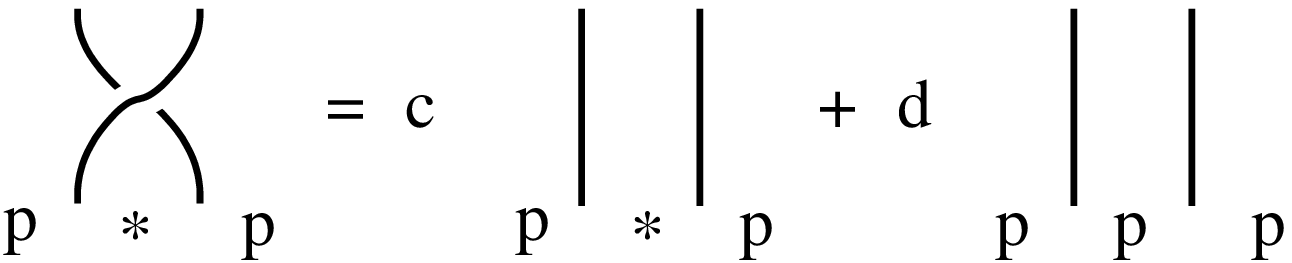}
\end{equation}
which means that the elementary crossing $\sigma_i$ corresponds to a
linear transformation which takes any string whose $i\th$ through 
$(i+2)\th$ symbols are $p*p$ to the coefficient $c$ times the same
string plus the coefficient $d$ times the same string with the $*$ at
the $(i+1)\th$ position replaced by $p$. (As shown in figure 9, the
$i\th$ crossing is over the $(i+1)\th$ symbol.) To compute the linear
transformation that the representation of a given braid applies to a
given string of symbols, one can write the symbols across the base of
the braid, and then apply rules of the form \ref{picto} until all the
crossings are removed, and all that remains are various coefficients
for different strings to be written across the base of a set of
straight strands.

For compactness, we will use $(p \widehat{*} p) = c(p * p) + d(p p p)$
as a shorthand for equation \ref{picto}. In this notation, the
complete set of rules is as follows.

\begin{eqnarray}
\label{rules}
(* \widehat{p} p) & = & a(*pp) \nonumber \\
(* \widehat{p} *) & = & b(*p*) \nonumber \\
(p \widehat{*} p) & = & c(p*p) + d(ppp) \nonumber \\
(p \widehat{p} *) & = & a(pp*) \nonumber \\
(p \widehat{p} p) & = & d(p*p) + e(ppp),
\end{eqnarray}
where
\begin{eqnarray}
\label{constants}
a & = & -A^4 \nonumber \\
b & = & A^8 \nonumber \\
c & = & A^8 \tau^2 - A^4 \tau \nonumber \\
d & = & A^8 \tau^{3/2} + A^4 \tau^{3/2} \nonumber \\
e & = & A^8 \tau - A^4 \tau^2 \nonumber \\
A & = & e^{-i3 \pi/5} \nonumber \\
\tau & = & 2/(1+\sqrt{5}).
\end{eqnarray}

Using these rules we can calculate any matrix from the Fibonacci
representation of the braid group. Notice that this is a reducible
representation. These rules do not allow the rightmost symbol or
leftmost symbol of the string to change. Thus the vector space
decomposes into four invariant subspaces, namely the subspace spanned
by strings which begin and end with $p$, and the $* \ldots *$,
$p \ldots *$, and $*\ldots p$ subspaces. As an example, we can use the
above rules to compute the action of $B_3$ on the $* \ldots p$
subspace. 

\begin{equation}
\label{examp}
\rho_{*p}^{(3)}(\sigma_1) = \left[ \begin{array}{cc}
b & 0 \\
0 & a
\end{array} \right]
\begin{array}{c}
\textrm{$*$p$*$p} \\
\textrm{$*$ppp}
\end{array}
\quad \quad \quad
\rho_{*p}^{(3)}(\sigma_2) = \left[ \begin{array}{cc}
c & d \\
d & e
\end{array} \right]
\begin{array}{c}
\textrm{$*$p$*$p} \\
\textrm{$*$ppp}
\end{array}
\end{equation}

In appendix \ref{proof} we prove that the Jones polynomial evaluated
at $t = e^{i 2 \pi/5}$ can be obtained as a weighted trace of the Fibonacci
representation over the $* \ldots *$ and $*\ldots p$ subspaces.

\section{Computing the Jones Polynomial in DQC1}
\label{containment}

As mentioned previously, the Fibonacci representation acts on the
vector space of formal linear combinations of strings of $p$ and $*$
symbols in which no two $*$ symbols are adjacent. The set of length
$n$ strings of this type, $P_n$, has $f_{n+2}$ elements, where $f_n$ is
the $n\th$ Fibonacci number: $f_1 = 1$, $f_2 = 1$, $f_3 = 2$, and so
on. As shown in appendix \ref{bijective}, one can construct a
bijective correspondence between these strings and the integers from
$0$ to $f_{n+2}-1$ as follows. If we think of $*$ as $1$ and $p$ as
$0$, then with a string $s_n s_{n-1} \ldots s_1$ we associate the
integer
\begin{equation}
\label{correspondence}
z(s) = \sum_{i=1}^n s_i f_{i+1}. 
\end{equation}
This is known as the Zeckendorf representation.

Representing integers as bitstrings by the usual method of place
value, we thus have a correspondence between the elements of $P_n$ and
the bitstrings of length $b = \lceil \log_2(f_{n+2}) \rceil$. This
correspondence will be a key element in computing the Jones polynomial
with a one clean qubit machine. Using a one clean qubit machine, one
can compute the trace of a unitary over the entire Hilbert space of
$2^n$ bitstrings. Using CNOT gates as above, one can also compute with
polynomial overhead the trace over a subspace whose dimension is a
polynomially large fraction of the dimension of the entire Hilbert
space. However, it is probably not possible in general for a one clean
qubit computer to compute the trace over subspaces whose dimension is
an exponentially small fraction of the dimension of the total Hilbert
space. For this reason, directly mapping the strings of $p$ and $*$
symbols to strings of $1$ and $0$ will not work. In contrast, the
correspondence described in equation \ref{correspondence} maps $P_n$
to a subspace whose dimension is at least half the dimension of the
full $2^b$-dimensional Hilbert space.

In outline, the DQC1 algorithm for computing the Jones polynomial works
as follows. Using the results described in section \ref{DQC1}, we will
think of the quantum circuit as acting on $b$ maximally mixed qubits
plus $O(1)$ clean qubits.  Thinking in terms of the computational
basis, we can say that the first $b$ qubits are in a uniform
probabilistic mixture of the $2^b$ classical bitstring states. By
equation \ref{correspondence}, most of these bitstrings correspond to
elements of $P_n$. In the Fibonacci representation, an elementary
crossing on strands $i$ and $i-1$ corresponds to a linear
transformation which can only change the value of the $i\th$ symbol in
the string of $p$'s and $*$'s. The coefficients for changing this
symbol or leaving it fixed depend only on the two neighboring
symbols. Thus, to simulate this linear transformation, we will use a
quantum circuit which extracts the $(i-1)\th$, $i\th$, and $(i+1)\th$
symbols from their bitstring encoding, writes them into an ancilla
register while erasing them from the bitstring encoding, performs the
unitary transformation prescribed by equation \ref{rules} on the
ancillas, and then transfers this symbol back into the bitstring
encoding while erasing it from the ancilla register. Constructing one
such circuit for each crossing, multiplying them together, and
performing DQC1 trace-estimation yields an approximation to the Jones
polynomial.

Performing the linear transformation demanded by equation \ref{rules}
on the ancilla register can be done easily by invoking gate set
universality (\emph{cf.} Solovay-Kitaev theorem \cite{Nielsen}) since
it is just a three-qubit unitary operation. The harder steps are transferring
the symbol values from the bitstring encoding to the ancilla register and
back.

It may be difficult to extract an arbitrary symbol from
the bitstring encoding. However, it is relatively easy to extract the
leftmost ``most significant'' symbol, which determines whether the
Fibonacci number $f_n$ is present in the sum shown in equation
\ref{correspondence}. This is because, for a string $s$ of length $n$,
$z(s) \geq f_{n-1}$ if and only if the leftmost symbol is $*$. Thus,
starting with a clean $\ket{0}$ ancilla qubit, one can transfer the
value of the leftmost symbol into the ancilla as follows. First,
check whether $z(s)$ (as represented by a bitstring using place value)
is $\geq f_{n-1}$. If so flip the ancilla qubit. Then, conditioned on
the value of the ancilla qubit, subtract $f_{n-1}$ from the
bitstring. (The subtraction will be done modulo $2^b$ for
reversibility.) 

Any classical circuit can be made reversible with only constant
overhead. It thus corresponds to a unitary matrix which permutes the
computational basis. This is the standard way of implementing classical
algorithms on a quantum computer\cite{Nielsen}. However, the resulting
reversible circuit may require clean ancilla qubits as work space in
order to be implemented efficiently. For a reversible circuit to be
implementable on a one clean qubit computer, it must be efficiently
implementable using at most logarithmically many clean
ancillas. Fortunately, the basic operations of arithmetic and
comparison for integers can all be done classically by NC1 circuits
\cite{Wegener}. NC1 is the complexity class for problems solvable by
classical circuits of logarithmic depth. As shown in \cite{Ambainis},
any classical NC1 circuit can be converted into a reversible circuit
using only three clean ancillas. This is a consequence of Barrington's
theorem. Thus, the process described above for extracting the leftmost
symbol can be done efficiently in DQC1.

More specifically, Krapchenko's algorithm for adding two $n$-bit
numbers has depth $\lceil \log n \rceil + O(\sqrt{\log n})$
\cite{Wegener}. A lower bound of depth $\log n$ is also known, so this
is essentially optimal \cite{Wegener}. Barrington's construction
\cite{Barrington} yields a sequence of $2^{2 d}$ gates on 3 clean
ancilla qubits \cite{Ambainis} to simulate a circuit of depth
$d$. Thus we obtain an addition circuit which has quadratic size (up
to a subpolynomial factor). Subtraction can be obtained analogously,
and one can determine whether $a \geq b$ can be done by subtracting
$a$ from $b$ and looking at whether the result is negative.

Although the construction based on Barrington's theorem has polynomial
overhead and is thus sufficient for our purposes, it seems worth
noting that it is possible to achieve better efficiency. As shown by
Draper \cite{Draper}, there exist ancilla-free  
quantum circuits for performing addition and subtraction, which
succeed with high probability and have nearly linear
size. Specifically, one can add or subtract a hardcoded number $a$
into an $n$-qubit register $\ket{x}$ modulo $2^n$ by performing
quantum Fourier transform, followed by $O(n^2)$ controlled-rotations,
followed by an inverse quantum Fourier transform. Furthermore, using
approximate quantum Fourier transforms\cite{Coppersmith, Barenco},
\cite{Draper} describes an approximate version of the circuit, which,
for any value of parameter $m$, uses a total of only $O(m n \log n)$
gates\footnote{A linear-size quantum circuit for exact ancilla-free
  addition is known, but it does not generalize easily to the case of
  hardcoded summands \cite{Cuccaro}.} to produce an output having an
inner product with $\ket{x+a \mod 2^n}$ of $1-O(2^{-m})$.

Because they operate modulo $2^n$, Draper's quantum circuits for
addition and subtraction do not immediately yield fast ancilla-free quantum
circuits for comparison, unlike the classical case. Instead, start with
an $n$-bit number $x$ and then introduce a single clean ancilla qubit
initialized to $\ket{0}$. Then subtract an $n$-bit hardcoded number
$a$ from this register modulo $2^{n+1}$. If $a > x$ then the result
will wrap around into the range $[2^n,2^{n+1}-1]$, in which case the
leading bit will be 1. If $a \leq x$ then the result will be in the
range $[0,2^n-1]$. After copying the result of this leading qubit and
uncomputing the subtraction, the comparison is
complete. Alternatively, one could use the linear size quantum 
comparison circuit devised by Takahashi and Kunihiro, which uses $n$
uninitialized ancillas but no clean ancillas\cite{Takahashi}.

\begin{figure}
\[
\begin{array}{cccc|cccccc}
1 & 2 & 3 & 5 & 13 & 8 & 5 & 3 & 2 & 1 \\
* & p & p & * & p  & p & * & p & p & p
\end{array}
\quad \leftrightarrow \quad (6,5)
\]
\caption{\label{pairs} Here we make a correspondence between strings
  of $p$ and $*$ symbols and ordered pairs of integers. The string of
  9 symbols is split into substrings of length 4 and 5, and each one
  is used to compute an integer by adding the $(i+1)\th$ Fibonacci
  number if $*$ appears in the $i\th$ place. Note the two strings are
  read in different directions.}
\end{figure}

Unfortunately, most crossings in a given braid will not be acting on
the leftmost strand. However, we can reduce the problem of extracting
a general symbol to the problem of extracting the leftmost
symbol. Rather than using equation \ref{correspondence} to make a
correspondence between a string from $P_n$ and a single integer, we 
can split the string at some chosen point, and use equation
\ref{correspondence} on each piece to make a correspondence between
elements of $P_n$ and ordered pairs of integers, as shown in figure
\ref{pairs}. To extract the $i\th$ symbol, we thus convert encoding
\ref{correspondence} to the encoding where the string is split between
the $i\th$ and $(i-1)\th$ symbols, so that one only needs to extract
the leftmost symbol of the second string. Like equation
\ref{correspondence}, this is also an efficient encoding, in which the
encoded bitstrings form a large fraction of all possible bitstrings.

To convert encoding \ref{correspondence} to a split encoding with the
split at an arbitrary point, we can move the split rightward by one
symbol at a time. To introduce a split between the leftmost and
second-to-leftmost symbols, one must extract the leftmost symbol as
described above. To move the split one symbol to the right, one must
extract the leftmost symbol from the right string, and if it is $*$
then add the corresponding Fibonacci number to the left string. This
is again a procedure of addition, subtraction, and comparison of
integers. Note that the computation of Fibonacci
numbers in NC1 is not necessary, as these can be hardcoded into the
circuits. Moving the split back to the left works analogously. As
crossings of different pairs of strands are being simulated, the split
is moved to the place that it is needed. At the end it is moved all
the way leftward and eliminated, leaving a superposition of bitstrings
in the original encoding, which have the correct coefficients
determined by the Fibonacci representation of the given braid.

Lastly, we must consider the weighting in the trace, as described by
equation \ref{weightdef}. Instead of weight $W_s$, we will use
$W_s/\phi$ so that the possible weights are 1 and $1/\phi$ both of
which are $\leq 1$. We can impose any weight $\leq 1$ by doing a
controlled rotation on an extra qubit. The CNOT trick for simulating
a clean qubit which was described in section \ref{DQC1} can be viewed
as a special case of this. All strings in which that qubit takes the
value $\ket{1}$ have weight zero, as imposed by a $\pi/2$ rotation on
the extra qubit. Because none of the weights are smaller than
$1/\phi$, the weighting will cause only a constant overhead
in the number of measurements needed to get a given precision.

\section{DQC1-hardness of Jones Polynomials}

We will prove DQC1-hardness of the problem of estimating the Jones
polynomial of the trace closure of a braid by a reduction from the
problem of estimating the trace of a quantum circuit. To do this, we
will specify an encoding, that is, a map $\eta:Q_n \to S_m$
from the set $Q_n$ of strings of $p$ and $*$ symbols which start with
$*$ and have no two $*$ symbols in a row, to $S_m$, the set of
bitstrings of length $m$. For a given quantum circuit, we will
construct a braid whose Fibonacci representation implements the
corresponding unitary transformation on the encoded bits. The Jones
polynomial of the trace closure of this braid, which is the trace of
this representation, will equal the trace of the encoded quantum
circuit.

Unlike in section \ref{containment}, we will not use a one to one
encoding between bit strings and strings of $p$ and $*$ symbols. All
we require is that a sum over all strings of $p$ and $*$ symbols
corresponds to a sum over bitstrings in which each bitstring appears
an equal number of times. Equivalently, all bitstrings $b \in S_m$
must have a preimage $\eta^{-1}(b)$ of the same size. This insures an
unbiased trace in which no bitstrings are overweighted. To achieve
this we can divide the symbol string into blocks of three symbols and
use the encoding

\begin{equation}
\label{codon}
\begin{array}{ccc}
\textrm{ppp} & \to & 0 \\
\textrm{p$*$p} & \to & 1 \\
\end{array}
\end{equation}

The strings other than ppp and p$*$p do not correspond to any bit
value. Since both the encoded 1 and the encoded 0 begin and end with
p, they can be preceded and followed by any allowable string. Thus,
changing an encoded 1 to an encoded zero does not change the number of
allowed strings of $p$ and $*$ consistent with that encoded
bitstring. Thus the condition that $|\eta^{-1}(b)|$ be independent of
$b$ is satisfied.

We would also like to know \emph{a priori} where in the string of p
and $*$ symbols a given bit is encoded. This way, when we need to
simulate a gate acting on a given bit, we would know which strands the
corresponding braid should act on. If we were to simply divide our
string of symbols into blocks of three and write down the
corresponding bit string (skipping every block which is not in one of
the two coding states ppp and p$*$p) then this would not be the
case. Thus, to encode $n$ bits, we will instead divide the string of
symbols into $n$ superblocks, each consisting of $c \log n$ blocks of
three for some constant $c$. To decode a superblock, scan it from left
to right until you reach either a ppp block or a p$*$p block. The
first such block encountered determines whether the superblock encodes
a 1 or a 0, according to equation \ref{codon}. Now imagine we choose a
string randomly from $Q_{3cn \log n}$. By choosing the constant
prefactor $c$ in our superblock size we can ensure that in the entire
string of $3cn \log n$ symbols,  the probability of there being any
noncoding superblock which contains neither a ppp block nor a p$*$p
block is polynomially small. If this is the case, then these noncoding
strings will contribute only a polynomially small additive error to
the estimate of the circuit trace, on par with the other sources of
error.

The gate set consisting of the CNOT, Hadamard, and $\pi/8$ gates is
known to be universal for BQP \cite{Nielsen}. Thus, it suffices to
consider the simulation of 1-qubit and 2-qubit gates. Furthermore, it
is sufficient to imagine the qubits arranged on a line and to allow
2-qubit gates to act only on neighboring qubits. This is because
qubits can always be brought into neighboring positions by applying a
series of SWAP gates to nearest neighbors. By our encoding a
unitary gate applied to qubits $i$ and $i+1$ will correspond to a
unitary transformation on symbols $i3c \log n$ through $(i+2)3c \log
n-1$. The essence of our reduction is to take each quantum gate
and represent it by a corresponding braid on logarithmically many 
symbols whose Fibonacci representation performs that gate on the
encoded qubits. 

Let's first consider the problem of simulating a gate on the first
pair of qubits, which are encoded in the leftmost two superblocks of
the symbol string. We'll subsequently consider the more difficult case
of operating on an arbitrary pair of neighboring encoded qubits. As
mentioned in section \ref{Fibonacci}, the Fibonacci representation 
$\rho_F^{(n)}$ is reducible. Let $\rho_{**}^{(n)}$ denote the
representation of the braid group $B_n$ defined by the action of
$\rho_F^{(n)}$ on the vector space spanned by strings which begin
and end with $*$. As shown in appendix \ref{density},
$\rho_{**}^{(n)}(B_n)$ taken modulo phase is a dense subgroup of
$SU(f_{n-1})$, and $\rho_{*p}^{(n)}(B_n)$ modulo phase is a dense
subgroup of $SU(f_n)$. 

In addition to being dense, the $**$ and $*$p blocks of the
Fibonacci representation can be controlled independently. This is a
consequence of the decoupling lemma, as discussed in appendix
\ref{density}. Thus, given a string of symbols beginning with $*$, and
any desired pair of unitaries on the corresponding $*$p and $**$
vector spaces, a braid can be constructed whose Fibonacci
representation approximates these unitaries to any desired level of
precision. However, the number of crossings necessary may in general
be large. The space spanned by strings of logarithmically many symbols
has only polynomial dimension. Thus, one might guess that the braid
needed to approximate a given pair of unitaries on the $*$p and $**$
vector spaces for logarithmically many symbols will have only
polynomially many crossings. It turns out that this guess is correct,
as we state formally below.
\begin{proposition}
\label{efficiency}
Given any pair of elements $U_{*p} \in SU(f_{k+1})$ and $U_{**} \in
SU(f_k)$, and any real parameter $\epsilon$, one can in polynomial
time find a braid $b \in B_k$ with $\mathrm{poly}(n,\log(1/\epsilon))$
crossings whose Fibonacci representation satisfies
$\| \rho_{*p}(b) - U_{*p} \| \leq \epsilon$ and 
$\| \rho_{**}(b) - U_{**} \| \leq \epsilon$, provided that $k = O(\log
n)$. By symmetry, the same holds when considering $\rho_{p*}$ rather
than $\rho_{*p}$.
\end{proposition}
Note that proposition \ref{efficiency} is a property of the Fibonacci
representation, not a generic consequence of density, since it is in
principle possible for the images of group generators in a dense
representation to lie exponentially close to some subgroup of the
corresponding unitary group. We prove this proposition in appendix
\ref{app_efficiency}.

With proposition \ref{efficiency} in hand, it is apparent that any
unitary gate on the first two encoded bits can be efficiently
performed. To similarly simulate gates on arbitrary pairs of
neighboring encoded qubits, we will need some way to unitarily bring a
$*$ symbol to a known location within logarithmic distance of the
relevant encoded qubits. This way, we ensure that we are acting in the
$*p$ or $**$ subspaces.

To move $*$ symbols to known locations we'll use an ``inchworm''
structure which brings a pair of $*$ symbols rightward to where they
are needed. Specifically, suppose we have a pair of superblocks which each
have a $*$ in their exact center. The presence of the left $*$ and the
density of $\rho_{*p}$ allow us to use proposition \ref{efficiency} to
unitarily move the right $*$ one superblock to the right by adding
polynomially many crossings to the braid. Then, the presence of the
right $*$ and the density of $\rho_{p*}$ allow us to similarly move
the left $*$ one superblock to the right, thus bringing it into the superblock
adjacent to the one which contains the right $*$. This is illustrated
in figure \ref{inchfig}. To move the inchworm to the left we use the
inverse operation.

\begin{figure}
\begin{center}
\includegraphics[width=0.35\textwidth]{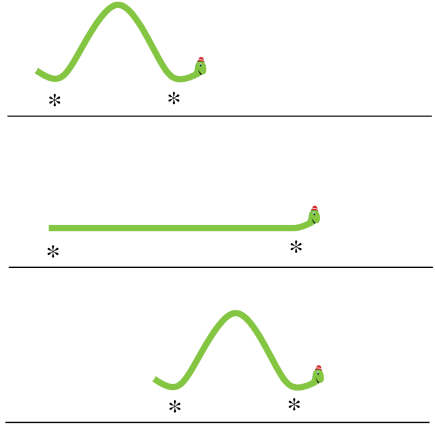}
\caption{\label{inchfig} This sequence of unitary steps is used to
  bring a $*$ symbol where it is needed in the symbol string to ensure
  density of the braid group representation. The presence of the left
  $*$ ensures density to allow the movement of the right $*$ by
  proposition \ref{efficiency}. Similarly, the presence of the right $*$
  allows the left $*$ to be moved.} 
\end{center}
\end{figure}

To simulate a given gate, one first uses the previously described
procedure to make the inchworm crawl to the superblocks just to the
left of the superblocks which encode the qubits on which the gate
acts. Then, by the density of $\rho_{*p}$ and proposition
\ref{efficiency}, the desired gate can be simulated using polynomially
many braid crossings.

To get this process started, the leftmost two superblocks must each contain
a $*$ at their center. This occurs with constant probability. The
strings in which this is not the case can be prevented from
contributing to the trace by a technique analogous to that used in
section \ref{DQC1} to simulate logarithmically many clean
ancillas. Namely, an extra encoded qubit can be conditionally flipped
if the first two superblocks do not both have $*$ symbols at their
center. This can always be done using proposition \ref{efficiency},
since the leftmost symbol in the string is always $*$, and the
$\rho_{*p}$ and $\rho_{**}$ representations are both dense.

It remains to specify the exact unitary operations which move the
inchworm. Suppose we have a current superblock and a target superblock. The
current superblock contains a $*$ in its center, and the target superblock is
the next superblock to the right or left. We wish to move the $*$ to the
center of the target superblock. To do this, we can select the smallest
segment around the center such that in each of these superblocks, the
segment is bordered on its left and right by p symbols. This segment
can then be swapped, as shown in figure \ref{swap}.

\begin{figure}
\begin{center}
\includegraphics[width=0.5\textwidth]{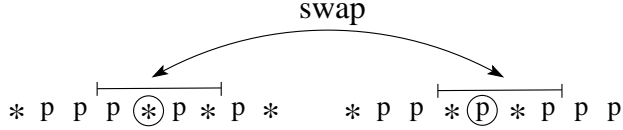}
\caption{\label{swap} This unitary procedure starts with a $*$ in the
  current superblock and brings it to the center of the target superblock.} 
\end{center}
\end{figure}

For some possible strings this procedure will not be well
defined. Specifically there may not be any segment which contains the
center and which is bordered by p symbols in both superblocks. On such
strings we define the operation to act as the identity. For random
strings, the probability of this decreases exponentially with the
superblock size. Thus, by choosing $c$ sufficiently large we can make this
negligible for the entire computation.

As the inchworm moves rightward, it leaves behind a trail. Due to the
swapping, the superblocks are not in their original state after the
inchworm has passed. However, because the operations are unitary, when
the inchworm moves back to the left, the modifications to the superblocks
get undone. Thus the inchworm can shuttle back and forth, moving where
it is needed to simulate each gate, always stopping just to the left
of the superblocks corresponding to the encoded qubits.

The only remaining detail to consider is that the trace appearing
in the Jones polynomial is weighted depending on whether the last
symbol is $p$ or $*$, whereas the DQC1-complete trace estimation
problem is for completely unweighted traces. This problem is easily
solved. Just introduce a single extra superblock at the end of the
string. After bringing the inchworm adjacent to the last superblock,
apply a unitary which performs a conditional rotation on the qubit
encoded by this superblock. The rotation will be by an angle so that
the inner product of the rotated qubit with its original state is
$1/\phi$ where $\phi$ is the golden ratio. This will be done only if
the last symbol is $p$. This exactly cancels out the weighting which
appears in the formula for the Jones polynomial, as described in
appendix \ref{proof}.

Thus, for appropriate $\epsilon$, approximating the Jones polynomial
of the trace closure of a braid to within $\pm \epsilon$ is
DQC1-hard.

\section{Conclusion}
\label{conclusion}

The preceding sections show that the problem of approximating the
Jones polynomial of the trace closure of a braid with $n$ strands and
$m$ crossings to within $\pm \epsilon$ at $t=e^{i 2 \pi/5}$ is a
DQC1-complete problem for appropriate $\epsilon$. The proofs are based
on the problem of evaluating the Markov trace of the Fibonacci
representation of a braid to $\frac{1}{\mathrm{poly}(n,m)}$
precision. By equation \ref{jones}, we see that this corresponds to
evaluating the Jones polynomial with $\pm
\frac{|D^{n-1}|}{\mathrm{poly}(n,m)}$ precision, where
$D=-A^2-A^{-2}=2\cos (6 \pi/5)$. Whereas approximating the Jones
polynomial of the plat closure of a braid was known\cite{Aharonov2} to
be BQP-complete, it was previously only known that the problem of
approximating the Jones polynomial of the trace closure of a braid was
in BQP. Understanding the complexity of approximating the Jones
polynomial of the trace closure of a braid to precision  
$\pm \frac{|D^{n-1}|}{\mathrm{poly}(n,m)}$ was posed as an open
problem in \cite{Aharonov1}. This paper shows that for $A=e^{-i 3
  \pi/5}$, this problem is DQC1-complete. Such a completeness result
improves our understanding of both the difficulty of the Jones
polynomial problem and the power one clean qubit computers by finding
an equivalence between the two.

It is generally believed that DQC1 is not contained in P and does not
contain all of BQP. The DQC1-completeness result shows that if this
belief is true, it implies that approximating the Jones polynomial of
the trace closure of a braid is not so easy that it can be done
classically in polynomial time, but is not so difficult as to be
BQP-hard.

To our knowledge, the problem of approximating the Jones polynomial of
the trace closure of a braid is one of only four known candidates for
classically intractable problems solvable on a one clean qubit
computer. The others are estimating the Pauli decomposition of the
unitary matrix corresponding to a polynomial-size quantum
circuit\footnote{This includes estimating the trace of the unitary as
  a special case.}, \cite{Knill, Shepherd}, estimating quadratically
signed weight enumerators\cite{qwgt}, and estimating average fidelity
decay of quantum maps\cite{decay1, decay2}.

\section{Acknowledgements}
The authors thank David Vogan, Pavel Etingof, Raymond Laflamme, Pawel
Wocjan, Sergei Bravyi, Wim van Dam, Aram Harrow, and Daniel Nagaj for
useful discussions, and an anonymous reviewer for helpful comments. PS
gratefully acknowledges support from the W. M. Keck foundation and
from the NSF under grant number CCF-0431787. SJ gratefully
acknowledges support from ARO/DTO's QuaCGR program and the use of MIT
Center for Theoretical Physics facilities as supported by the US
Department of Energy. Part of this work was completed while SJ was a
short term visitor at Perimeter Institute.

\appendix

\section{Jones Polynomials by Fibonacci Representation}
\label{proof}
For any braid $b \in B_n$ we will define $\ttr(b)$ by:
\begin{equation}
\label{tracedef}
\ttr(b) = \frac{1}{\phi f_n + f_{n-1}} \sum_{s \in Q_{n+1}}
W_s \mx
\end{equation}
We will use $\mid$ to denote a strand and $\nmid$ to denote multiple
strands of a braid (in this case $n$). $Q_{n+1}$ is the set of all
strings of $n+1$ $p$ and $*$ symbols which start with $*$ and contain
no two $*$ symbols in a row. The symbol
\[
\mx
\]
denotes the $s,s$ matrix element of the Fibonacci representation of
braid $b$. The weight $W_s$ is
\begin{equation}
\label{weightdef}
W_s = \left\{ \begin{array}{ll}
\phi & \textrm{if $s$ ends with $p$} \\
1    & \textrm{if $s$ ends with $*$}.
\end{array} \right.
\end{equation}
$\phi$ is the golden ratio $(1+\sqrt{5})/\sqrt{2}$. 

As discussed in \cite{Aharonov1}, the Jones polynomial of the trace
closure of a braid $b$ is given by
\begin{equation}
\label{jones}
V_{b^{\mathrm{tr}}}(A^{-4}) = (-A)^{3w(b^{\mathrm{tr}})} D^{n-1} 
\tr(\rho_A(b^{\mathrm{tr}})).
\end{equation}
$b^{\mathrm{tr}}$ is the link obtained by taking the trace
closure of braid $b$. $w(b^{\mathrm{tr}})$ is denotes the writhe of
the link $b^{\mathrm{tr}}$. For an oriented link, one assigns a
value of $+1$ to each crossing of the form
\includegraphics[width=0.2in]{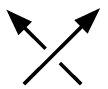}, and the value $-1$
to each crossing of the form
\includegraphics[width=0.2in]{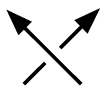}. The writhe of a link is
defined to be the sum of these values over all crossings. $D$ is
defined by $D = -A^2-A^{-2}$. $\rho_A: B_n \to \mathrm{TL}_n(D)$ is a
representation from the braid group to the Temperley-Lieb algebra with
parameter $D$. Specifically, 
\begin{equation}
\label{rhoa}
\rho_A(\sigma_i) = A E_i + A^{-1} \id
\end{equation}
where $E_1 \ldots E_n$ are the generators of $\mathrm{TL}_n(D)$, which
satisfy the following relations.
\begin{eqnarray}
E_i E_j & = & E_j E_i \quad \textrm{ for } |i-j| > 1 \label{rel1} \\
E_i E_{i \pm 1} E_i & = & E_i \label{rel2} \\
E_i^2 & = & D E_i \label{rel3}
\end{eqnarray}
The Markov trace on $\mathrm{TL}_n(D)$ is a linear map $\tr:
\mathrm{TL}_n(D) \to \mathbb{C}$ which satisfies 
\begin{eqnarray}
\tr(\id) & = & 1 \label{mark1} \\
\tr(XY) & = & \tr(YX) \label{mark2} \\
\tr(X E_{n-1}) & = & \frac{1}{D} \tr(X') \label{mark3}
\end{eqnarray}
On the left hand side of equation \ref{mark3}, the trace is on
$\mathrm{TL}_n(D)$, and $X$ is an element of $\mathrm{TL}_n(D)$ not
containing $E_{n-1}$. On the right hand side of equation \ref{mark3},
the trace is on $\mathrm{TL}_{n-1}(D)$, and $X'$ is the element of
$\mathrm{TL}_{n-1}(D)$ which corresponds to $X$ in the obvious way
since $X$ does not contain $E_{n-1}$.

We'll show that the Fibonacci representation satisfies the properties
implied by equations \ref{rhoa}, \ref{rel1}, \ref{rel2}, and
\ref{rel3}. We'll also show that $\ttr$ on the Fibonacci
representation satisfies the properties corresponding to \ref{mark1},
\ref{mark2}, and \ref{mark3}. It was shown in \cite{Aharonov1} that
properties \ref{mark1}, \ref{mark2}, and \ref{mark3}, along with
linearity, uniquely determine the map $\tr$. It will thus follow that
$\ttr(\rho_F^{(n)}(b)) = \tr(\rho_A(b))$, which proves that
the Jones polynomial is obtained from the trace $\ttr$ of
the Fibonacci representation after multiplying by the appropriate
powers of $D$ and $(-A)$ as shown in equation \ref{jones}. Since these
powers are trivial to compute, the problem of approximating the Jones
polynomial at $A = e^{-i 3 \pi/5}$ reduces to the problem of computing
this trace.

\newcommand{\symbox}[1]
{\begin{array}{l} \includegraphics[width=0.5in]{#1.eps} \end{array}}

$\ttr$ is equal to the ordinary matrix trace on the subspace of
strings ending in $*$ plus $\phi$ times the matrix trace on the
subspace of strings ending in p. Thus the fact that the matrix trace
satisfies property \ref{mark2} immediately implies that $\ttr$ does
too. Furthermore, since the dimensions of these subspaces are
$f_{n-1}$ and $f_n$ respectively, we see from equation \ref{tracedef}
that $\ttr(\id) = 1$. To address property \ref{mark3}, we'll
first show that
\begin{equation}
\label{toshow}
\ttr \left( \symbox{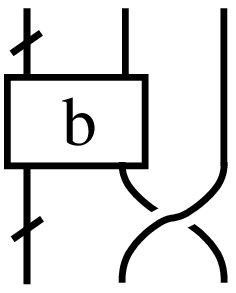} \right) =
  \frac{1}{\delta} \ttr \left( \begin{array}{l}
  \includegraphics[width=0.3in]{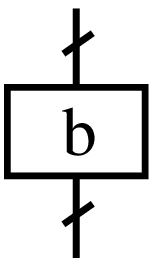} \end{array} \right)
\end{equation}
for some constant $\delta$ which we will calculate. We will then use
equation \ref{rhoa} to relate $\delta$ to $D$.

Using the definition of $\ttr$ we obtain
\begin{eqnarray*}
\ttr \left( \symbox{braidbox3} \right) & = &\frac{1}{f_n \phi + f_{n-1}}
\left[ \phi \sum_{s \in Q_{n-2}} \symbox{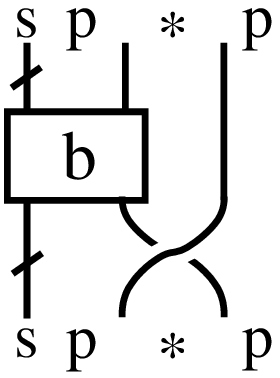} + \phi \sum_{s \in
    Q_{n-2}} \symbox{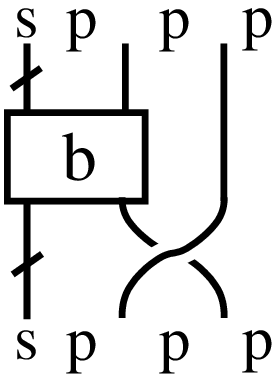} \right.\\
 & & \left. + \sum_{s \in Q_{n-2}} \symbox{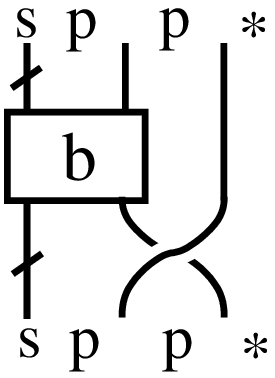} + \sum_{s \in
    Q'_{n-2}} \symbox{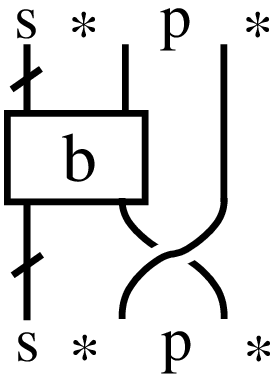} + \phi \sum_{s \in Q'_{n-2}} \symbox{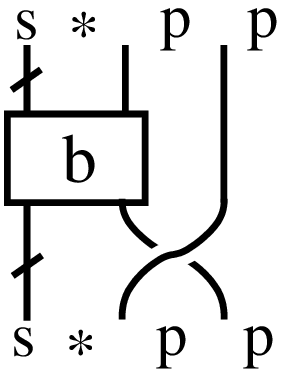} \right]
\end{eqnarray*}
where $Q'_{n-2}$ is the set of length $n-2$ strings of $*$ and $p$
symbols which begin with $*$, end with $p$, and have no two $*$
symbols in a row.

Next we expand according to the braiding rules described in equations
\ref{picto} and \ref{rules}.
\begin{eqnarray*}
& = &\frac{1}{f_n \phi + f_{n-1}}
\left[ \sum_{s \in Q_{n-2}} \left( \phi c \symbox{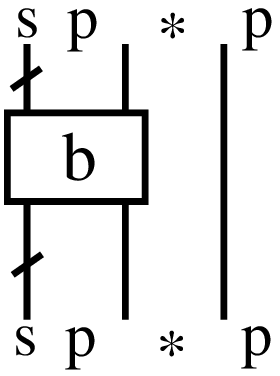} + \phi
  e \symbox{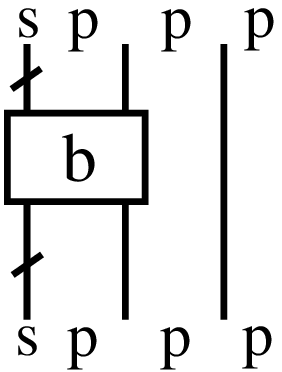} + a \symbox{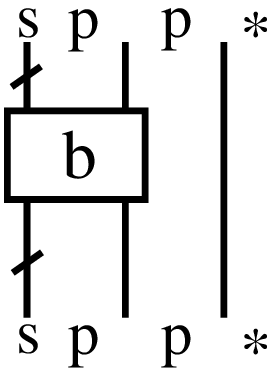} \right) \right.\\
 & & \left. + \sum_{s \in Q'_{n-2}} \left( b \symbox{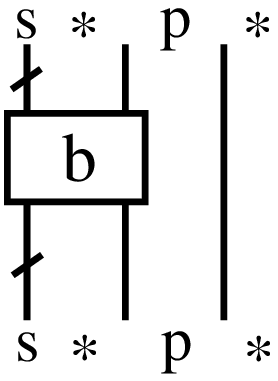} +
  \phi a \symbox{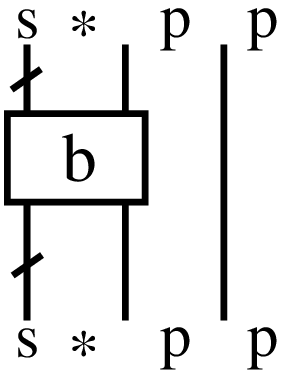} \right) \right]
\end{eqnarray*}
We know that matrix elements in which differing string symbols are
separated by unbraided strands will be zero. To obtain the preceding
expression we have omitted such terms.  Simplifying yields
\newcommand{\smallbox}[1]
{\begin{array}{l} \includegraphics[width=0.35in]{#1.eps} \end{array}}
\[
= \frac{1}{f_n + \phi f_{n-1}} \left[ \sum_{s \in Q_{n-2}} \left( \phi
  c \smallbox{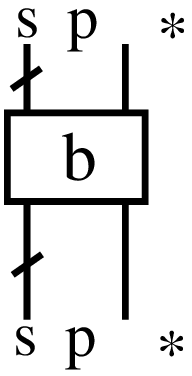} + (\phi e + a) \smallbox{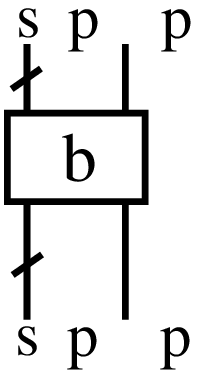} \right) + \sum_{s \in
  Q'_{n-2}} (b+ \phi a) \smallbox{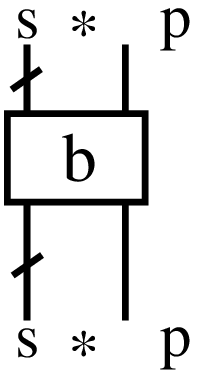} \right].
\]
By the definitions of $A$, $a$, $b$, and $e$, given in equation
\ref{constants}, we see that $\phi e + a = b + \phi a$. Thus the above
expression simplifies to
\newcommand{\minibox}[1]
{\begin{array}{l} \includegraphics[width=0.3in]{#1.eps} \end{array}}
\[
= \frac{1}{f_n \phi + f_{n-1}} \left[ \sum_{s \in Q'_{n-1}} \phi c
  \minibox{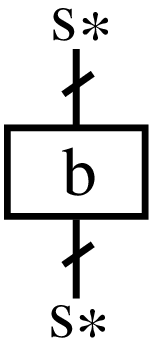} + \sum_{s \in Q_{n-1}} (\phi e + a) \minibox{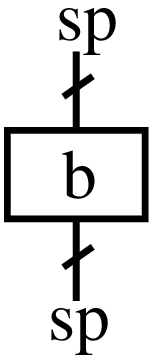} \right]
\]
Now we just need to show that
\begin{equation}
\label{delta1}
\frac{\phi c}{f_n \phi + f_{n-1}} = 
\frac{1}{\delta}\frac{1}{f_{n-1} \phi + f_{n-2}}
\end{equation}
and
\begin{equation}
\label{delta2}
\frac{\phi e + a}{f_n \phi + f_{n-1}} =
\frac{1}{\delta} \frac{\phi}{f_{n-1} \phi + f_{n-2}}.
\end{equation}
The Fibonacci numbers have the property
\[
\frac{f_n \phi + f_{n-1}}{f_{n-1} \phi + f_{n-2}} = \phi
\]
for all $n$. Thus equations \ref{delta1} and \ref{delta2} are
equivalent to
\begin{equation}
\label{equiv1}
\phi c = \frac{1}{\delta} \phi
\end{equation}
and
\begin{equation}
\label{equiv2}
\phi e + a = \frac{1}{\delta} \phi^2
\end{equation}
respectively. For $A=e^{-i 3 \pi/5}$ these both yield $\delta =
A-1$. Hence
\[
\ttr \left( \symbox{braidbox3} \right) = \frac{1}{\delta}
\frac{1}{f_{n-1} \phi + f_{n-2}} \left[ \sum_{s \in Q'_{n-1}}
  \minibox{s} + \sum_{s \in Q_{n-1}} \phi \minibox{p} \right]
\]
\[
= \frac{1}{\delta} \ttr(b),
\]
thus confirming equation \ref{toshow}. 

Now we calculate $D$ from $\delta$. Solving \ref{rhoa} for $E_i$
yields
\begin{equation}
\label{ei}
E_i = A^{-1} \rho_A(\sigma_i) - A^{-2} \id
\end{equation}
Substituting this into \ref{mark3} yields
\[
\tr(X (A^{-1} \rho_A(\sigma_i)-A^{-2} \id)) = \frac{1}{D} \tr(X)
\]
\[
\Rightarrow A^{-1} \tr(X \rho_A(\sigma_i)) - A^{-2} \tr(X) =
\frac{1}{D} \tr(X).
\]
Comparison to our relation $\tr(X \rho_A(\sigma_i)) = \frac{1}{\delta}
\tr(X)$ yields
\[
A^{-1} \frac{1}{\delta} - A^{-2} = \frac{1}{D}.
\]
Solving for $D$ and substituting in $A=e^{-i 3 \pi/5}$ yields
\[
D = \phi.
\]
This is also equal to $-A^2-A^{-2}$ consistent with the usage
elsewhere.

Thus we have shown that $\ttr$ has all the necessary properties. We
will next show that the image of the representation $\rho_F$ of the
braid group $B_n$ also forms a representation of the Temperley-Lieb
algebra $TL_n(D)$. Specifically, $E_i$ is represented by
\begin{equation}
\label{rep}
E_i \to A^{-1} \rho_F^{(n)}(\sigma_i) - A^{-2} \id.
\end{equation}
To show that this is a representation of $TL_n(D)$ we must show that
the matrices described in equation \ref{rep} satisfy the properties
\ref{rel1}, \ref{rel2}, and \ref{rel3}. By the theorem of
\cite{Aharonov1} which shows that a Markov trace on any representation
of the Temperley-Lieb algebra yields the Jones polynomial, it will
follow that the trace of the Fibonacci representation yields the Jones
polynomial.

Since $\rho_F$ is a representation of the braid group and $\sigma_i
\sigma_j = \sigma_j \sigma_i$ for $|i-j| > 1$, it immediately follows
that the matrices described in equation \ref{rep} satisfy condition
\ref{rel1}. Next, we'll consider condition \ref{rel3}. By inspection
of the Fibonacci representation as given by equation 
\ref{rules}, we see that by appropriately ordering the
basis\footnote{We will have to choose different orderings for
  different $\sigma_i$'s.} we can bring $\rho_A(\sigma_i)$ into block
diagonal form, where each block is one of the following $1 \times 1$
or $2 \times 2$ possibilities.
\[
\left[ a \right] \quad \left[ b \right] \quad 
\left[ \begin{array}{cc}
c & d \\
d & e
\end{array} \right]
\]
Thus, by equation \ref{rep}, it suffices to show that
\[
\left( A^{-1} \left[ \begin{array}{cc}
c & d \\
d & e
\end{array} \right]
- A^{-2} \left[ \begin{array}{cc}
1 & 0 \\
0 & 1
\end{array} \right] \right)^2 = D \left( A^{-1}
\left[ \begin{array}{cc}
c & d \\
d & e
\end{array} \right]
- A^{-2} \left[ \begin{array}{cc}
1 & 0 \\
0 & 1
\end{array} \right] \right),
\]
\[
(A^{-1} a - A^{-2})^2 = D (A^{-1} a - A^{-2} ),
\]
and
\[
(A^{-1} b - A^{-2})^2 = D (A^{-1} b - A^{-2}),
\]
\emph{i.e.} each of the blocks square to $D$ times themselves. These
properties are confirmed by direct calculation.

Now all that remains is to check that the correspondence \ref{rep}
satisfies property \ref{rel2}. Using the rules described in equation
\ref{rules} we have
\[
\rho_F^{(3)}(\sigma_1) = \left[ \begin{array}{cccccccc}
b & 0 &   &   &   &   &   &   \\
0 & a &   &   &   &   &   &   \\
  &   & a &   &   &   &   &   \\
  &   &   & e & 0 & d &   &   \\
  &   &   & 0 & a & 0 &   &   \\
  &   &   & d & 0 & c &   &   \\
  &   &   &   &   &   & e & d \\
  &   &   &   &   &   & d & c
\end{array} \right] 
\begin{array}{c}
\textrm{$*$p$*$p} \\
\textrm{$*$ppp} \\
\textrm{$*$pp$*$} \\
\textrm{pppp} \\
\textrm{pp$*$p} \\
\textrm{p$*$pp} \\
\textrm{ppp$*$} \\
\textrm{p$*$p$*$} \\
\end{array}
\quad \quad
\rho_F^{(3)}(\sigma_2) = \left[ \begin{array}{cccccccc}
c & d &   &   &   &   &   &   \\
d & e &   &   &   &   &   &   \\
  &   & a &   &   &   &   &   \\
  &   &   & e & d & 0 &   &   \\
  &   &   & d & c & 0 &   &   \\
  &   &   & 0 & 0 & a &   &   \\
  &   &   &   &   &   & a & 0 \\
  &   &   &   &   &   & 0 & b
\end{array} \right]
\begin{array}{c}
\textrm{$*$p$*$p} \\
\textrm{$*$ppp} \\
\textrm{$*$pp$*$} \\
\textrm{pppp} \\
\textrm{pp$*$p} \\
\textrm{p$*$pp} \\
\textrm{ppp$*$} \\
\textrm{p$*$p$*$} \\
\end{array}
\]
(Here we have considered all four subspaces unlike in equation
\ref{examp}.) Substituting this into equation \ref{rep} yields
matrices which satisfy condition \ref{rel3}. It follows that equation
\ref{rep} yields a representation of the Temperley-Lieb algebra. This
completes the proof that
\[
V_{b^{\mathrm{tr}}}(A^{-4}) = (-A)^{3w(b^{\mathrm{tr}})} D^{n-1} 
\ttr(\rho_F^{(n)}(b^{\mathrm{tr}}))
\]
for $A=e^{-i 3 \pi/5}$.

\section{Density of the Fibonacci representation}
\label{density}

In this appendix we will show that $\rho_{**}^{(n)}(B_n)$ is a dense
subgroup of $SU(f_{n-1})$ modulo phase, and that
$\rho_{*p}^{(n)}(B_n)$ and $\rho_{p*}^{(n)}(B_n)$ are dense subgroups
of $SU(f_n)$ modulo phase. Similar results regarding the path model
representation of the braid group were proven in \cite{Aharonov2}. Our
proofs will use many of the techniques introduced there.

We'll first show that $\rho_{**}^{(4)}(B_4)$ modulo phase is a dense
subgroup of $SU(2)$. We can then use the bridge lemma from
\cite{Aharonov2} to extend the result to arbitrary $n$.

\begin{proposition} \label{SU2} $\rho_{**}^{(4)}(B_4)$ modulo phase is
  a dense subgroup of $SU(2)$.
\end{proposition}

\begin{proof}
Using equation \ref{rules} we have:
\[
\rho_{**}^{(4)}(\sigma_1) = \rho_{**}^{(4)}(\sigma_3) =
\left[ \begin{array}{cc}
b & 0 \\
0 & a
\end{array} \right] \begin{array}{c}
\textrm{$*$p$*$p$*$} \\
\textrm{$*$ppp$*$}
\end{array}
\quad \quad
\rho_{**}^{(4)}(\sigma_2) =
\left[ \begin{array}{cc}
c & d \\
d & e
\end{array} \right] \begin{array}{c}
\textrm{$*$p$*$p$*$} \\
\textrm{$*$ppp$*$}
\end{array}
\]
We do not care about global phase so we will take
\[
\rho_{**}^{(n)}(\sigma_i) \quad \to \quad \frac{1}{ \left( \det
  \rho_{**}^{(n)}(\sigma_i) \right)^{1/f_{n-1}}}
\ \rho_{**}^{(n)}(\sigma_i) 
\]
to project into $SU(f_{n-1})$. Thus we must show the group $\langle
A,B \rangle$ generated by
\begin{equation}
\label{AB}
A = \frac{1}{\sqrt{ab}} \left[ \begin{array}{cc}
b & 0 \\
0 & a
\end{array} \right] \quad \quad
B = \frac{1}{\sqrt{ce-d^2}} \left[ \begin{array}{cc}
c & d \\
d & e
\end{array} \right]
\end{equation}
is a dense subgroup of $SU(2)$. To do this we will use the well known
surjective homomorphism $\phi:SU(2) \to SO(3)$ whose kernel is $\{ \pm
  \id \}$ (\emph{cf.} \cite{Artin}, pg. 276). A general element of
  $SU(2)$ can be written as
\[
\cos\left( \frac{\theta}{2} \right) \id + i \sin \left( \frac{\theta}{2}
\right) \left[ x \sigma_x + y \sigma_y + z \sigma_z \right]
\]
where $\sigma_x$, $\sigma_y$, $\sigma_z$ are the Pauli matrices,
and $x$, $y$, $z$ are real numbers satisfying $x^2+y^2+z^2=1$. $\phi$
maps this element to the rotation by angle $\theta$ about the axis 
\[
\vec{x} = \left[ \begin{array}{c} 
x \\ 
y \\ 
z \end{array} \right].
\]

Using equations \ref{AB} and \ref{constants}, one finds that $\phi(A)$
and $\phi(B)$ are both rotations by $7 \pi /5$. These rotations are
about different axes which are separated by angle
\[
\theta_{12} = \cos^{-1} (2-\sqrt{5})
\simeq
1.8091137886\ldots
\]
To show that $\rho_{**}^{(4)}(B_4)$ modulo phase is a dense subgroup
of $SU(2)$ it suffices to show that $\phi(A)$ and $\phi(B)$ generate
a dense subgroup of $SO(3)$. To do this we take advantage of the fact
that the finite subgroups of $SO(3)$ are completely known.

\begin{theorem}$\mathrm{(\cite{Artin} \ pg. \ 184)}$
Every finite subgroup of $SO(3)$ is one of the following:
\begin{itemize}
\item[] $C_k$: the cyclic group of order $k$
\item[] $D_k$: the dihedral group of order $k$
\item[] $T$: the tetrahedral group (order 12)
\item[] $O$: the octahedral group (order 24)
\item[] $I$: the icosahedral group (order 60)
\end{itemize}
\end{theorem}
The infinite proper subgroups of $SO(3)$ are all isomorphic to $O(2)$
or $SO(2)$. Thus, since $\phi(A)$ and $\phi(B)$ are rotations about
different axes, $\langle \phi(A), \phi(B) \rangle$ can only be
$SO(3)$ or a finite subgroup of $SO(3)$. If we can show that $\langle
\phi(A), \phi(B) \rangle$ is not contained in any of the finite
subgroups of $SO(3)$ then we are done.

Since $\phi(A)$ and $\phi(B)$ are rotations about different axes we
know that $\langle \phi(A), \phi(B) \rangle$ is not $C_k$ or
$D_k$. Next, we note that $R=\phi(A)^5 \phi(B)^5$ is a rotation by $2
\theta_{12}$. By direct calculation, $2 \theta_{12}$ is not an integer
multiple of $2 \pi / k$ for $k = 1,2,3,4,$ or 5. Thus $R$ has order
greater than 5. As mentioned on pg. 262 of \cite{Jones2}, $T$, $O$,
and $I$ do not have any elements of order greater than 5. Thus,
$\langle \phi(A), \phi(B) \rangle$ is not contained in $C$, $O$, or
$I$, which completes the proof. Alternatively, using more arithmetic
and less group theory, we can see that $2 \theta_{12}$ is not any
integer multiple of $2 \pi / k$ for any $k \leq 30$, thus $R$ cannot
be in $T$, $O$, or $I$ since its order does not divide the order of
any of these groups.
\end{proof}

Next we'll consider $\rho_{**}^{(n)}$ for larger $n$. These will be
matrices acting on the strings of length $n+1$. These can be divided
into those which end in pp$*$ and those which end in $*$p$*$. The
space upon which $\rho_{**}^{(n)}$ acts can correspondingly be divided
into two subspaces which are the span of these two sets of
strings. From equation \ref{rules} we can see that
$\rho_{**}^{(n)}(\sigma_1)\ldots\rho_{**}^{(n)}(\sigma_{n-3})$ will
leave these subspaces invariant. Thus if we order our basis to respect
this grouping of strings,
$\rho_{**}^{(n)}(\sigma_1)\ldots\rho_{**}^{(n)}(\sigma_{n-3})$ will
appear block-diagonal with a block corresponding to each of these
subspaces.

The possible prefixes of $*$p$*$ are all strings of length $n-2$ that
start with $*$ and end with p. Now consider the strings acted upon by
$\rho_{**}^{(n-2)}$. These have length $n-1$ and must end in $*$. The 
possible prefixes of this $*$ are all strings of length $n-2$ that
begin with $*$ and end with p. Thus these are in one to one
correspondence with the strings acted upon by $\rho_{**}^{(n)}$ that
end in $*$p$*$. Furthermore, since the rules \ref{rules} depend only
on the three symbols neighboring a given crossing, the block of 
$\rho_{**}^{(n)}(\sigma_1)\ldots\rho_{**}^{(n)}(\sigma_{n-3})$
corresponding to the $*$p$*$ subspace is exactly the same as 
$\rho_{**}^{(n-2)}(\sigma_1)\ldots\rho_{**}^{(n-2)}(\sigma_{n-3})$. By
a similar argument, the block of 
$\rho_{**}^{(n)}(\sigma_1)\ldots\rho_{**}^{(n)}(\sigma_{n-3})$
corresponding to the pp$*$ is exactly the same as 
$\rho_{**}^{(n-1)}(\sigma_1)\ldots\rho_{**}^{(n-1)}(\sigma_{n-3})$.

For any $n> 3$, $\rho_{**}^{(n)}(\sigma_{n-2})$ will not leave these
subspaces invariant. This is because the crossing $\sigma_{n-2}$ spans
the $(n-1)\th$ symbol. Thus if the $(n-2)\th$ and $n\th$ symbols are
$p$, then by equation \ref{rules}, $\rho_{**}^{(n)}$ can flip the
value of the $(n-1)\th$ symbol. The $n\th$ symbol is guaranteed to be
$p$, since the $(n+1)\th$ symbol is the last one and is therefore $*$
by definition. For any $n > 3$, the space acted upon by
$\rho_{**}^{(n)}(\sigma_{n-1})$ will include some strings in which the 
$(n-2)\th$ symbol is $p$. 

As an example, for five strands:
\[
\rho_{**}^{(5)}(\sigma_1) = \left[ \begin{array}{ccc}
b & 0 & 0 \\
0 & a & 0 \\
0 & 0 & a
\end{array} \right]
\begin{array}{c}
\textrm{$*$p$*$pp$*$} \\
\textrm{$*$pppp$*$} \\
\textrm{$*$pp$*$p$*$}
\end{array} \quad \quad
\rho_{**}^{(5)}(\sigma_2) = \left[ \begin{array}{ccc}
c & d & 0 \\
d & e & 0 \\
0 & 0 & a
\end{array} \right]
\begin{array}{c}
\textrm{$*$p$*$pp$*$} \\
\textrm{$*$pppp$*$} \\
\textrm{$*$pp$*$p$*$}
\end{array}
\]
\[
\rho_{**}^{(5)}(\sigma_3) = \left[ \begin{array}{ccc}
a & 0 & 0 \\
0 & e & d \\
0 & d & c
\end{array} \right]
\begin{array}{c}
\textrm{$*$p$*$pp$*$} \\
\textrm{$*$pppp$*$} \\
\textrm{$*$pp$*$p$*$}
\end{array} \quad \quad
\rho_{**}^{(5)}(\sigma_4) = \left[ \begin{array}{ccc}
a & 0 & 0 \\
0 & a & 0 \\
0 & 0 & b
\end{array} \right]
\begin{array}{c}
\textrm{$*$p$*$pp$*$} \\
\textrm{$*$pppp$*$} \\
\textrm{$*$pp$*$p$*$}
\end{array}
\]
We recognize the upper $2 \times 2$ blocks of
$\rho_{**}^{(5)}(\sigma_1)$, and $\rho_{**}^{(5)}(\sigma_2)$ from
equation \ref{examp}. The lower $1 \times 1$ block matches
$\rho_{**}^{(3)}(\sigma_1)$ and $\rho_{**}^{(3)}(\sigma_2)$, which are
both easily calculated to be $[a]$. $\rho_{**}^{(5)}(\sigma_3)$ mixes
these two subspaces.

We can now use the preceding observations about the recursive
structure of $\{ \rho_{**}^{(n)}| n = 4,5,6,7\ldots \}$ to show
inductively that $\rho_{**}^{(n)}(B_n)$ forms a dense subgroup of
$SU(f_{n-1})$ for all $n$. To perform the induction step we use the
bridge lemma and decoupling lemma from \cite{Aharonov2}.

\begin{lemma}[Bridge Lemma]
Let $C=A \oplus B$ where $A$ and $B$ are vector spaces with
$\dim{B} > \dim{A} \geq 1$. Let $W \in SU(C)$ be a linear
transformation which mixes the subspaces $A$ and $B$. Then the group
generated by $SU(A)$, $SU(B)$, and $W$ is dense in $SU(C)$.
\end{lemma}

\begin{lemma}[Decoupling Lemma]
Let $G$ be an infinite discrete group, and let $A$ and $B$ be two
vector spaces with $\dim(A) \neq \dim(B)$. Let $\rho_a: G \to SU(A)$
and $\rho_b: G \to SU(B)$ be homomorphisms such that $\rho_a(G)$ is
dense in $SU(A)$ and $\rho_b(G)$ is dense in $SU(B)$.  Then for any $U_a
\in SU(A)$ there exist a series of $G$-elements $\alpha_n$ such that
$\lim_{n \to \infty} \rho_a(\alpha_n) = U_a$ and $\lim_{n \to \infty}
\rho_b(\alpha_n) = \id$. Similarly, for any $U_b \in SU(B)$, there exists a
series $\beta_n \in G$ such that $\lim_{n \to \infty} \rho_a(\beta_n) = \id$
and $\lim_{n \to \infty} \rho_a(\beta_n) = U_b$.
\end{lemma}

With these in hand we can prove the main proposition of this
appendix.
\begin{proposition}
\label{**density}
For any $n \geq 3$, $\rho_{**}^{(n)}(B_n)$ modulo phase is a dense subgroup
of $SU(f_{n-1})$.
\end{proposition}

\begin{proof}
As mentioned previously, the proof will be inductive. The base cases are
$n=3$ and $n=4$. As mentioned previously, $\rho_{**}^{(3)}(\sigma_1) =
\rho_{**}^{(3)}(\sigma_2) = [a]$. Trivially, these generate a dense
subgroup of (indeed, all of) $SU(1) = \{ \id \}$ modulo phase. By
proposition \ref{SU2}, $\rho_{**}^{(4)}(\sigma_1)$, and
$\rho_{**}^{(4)}(\sigma_2)$ generate a dense subgroup of $SU(2)$
modulo phase. Now for induction assume that $\rho_{**}^{(n-1)}(B_{n-1})$ is
a dense subgroup of $SU(f_{n-2})$ and $\rho_{**}^{(n-2)}(B_{n-2})$ is
a dense subgroup of $SU(f_{n-3})$. As noted above, these correspond to
the upper and lower blocks of $\rho_{**}^{(n)}(\sigma_1) \ldots
\rho_{**}^{(n)}(\sigma_{n-2})$. Thus, by the decoupling lemma,
$\rho_{**}^{(n)}(B_n)$ contains an element arbitrarily close to $U
\oplus \id$ for any $U \in SU(f_{n-2})$ and an element arbitrarily
close to $\id \oplus U$ for any $U \in SU(f_{n-3})$. Since, as
observed above, $\rho_{**}^{(n)}(\sigma_{n-1})$ mixes these two
subspaces, the bridge lemma shows that $\rho_{**}^{(n)}(B_n)$ is dense
in $SU(f_{n-1})$.
\end{proof}

From this, the density of $\rho_{*p}^{(n)}$ and $\rho_{p*}^{(n)}$
easily follow.

\begin{corollary}
$\rho_{*p}^{(n)}(B_n)$ and $\rho_{p*}^{(n)}(B_n)$ are dense subgroups
  of $SU(f_n)$ modulo phase.
\end{corollary}

\begin{proof}
It is not hard to see that
\[
\begin{array}{rcl}
\rho_{*p}^{(n)}(\sigma_1) & = & \rho_{**}^{(n+1)}(\sigma_1) \\
 & \vdots & \\
\rho_{*p}^{(n)}(\sigma_{n-1}) & = & \rho_{**}^{(n+1)}(\sigma_{n-1})
\end{array}.
\]
As we saw in the proof of proposition \ref{**density},
$\rho_{**}^{(n+1)}(\sigma_n)$ is not necessary to obtain density in
$SU(f_n)$, that is, $\langle \rho_{**}^{(n+1)}(\sigma_1) , \ldots,
\rho_{**}^{(n+1)}(\sigma_{n-1}) \rangle$ is a dense subgroup of
$SU(f_n)$ modulo phase. Thus, the density of $\rho_{*p}^{(n)}$ in
$SU(f_n)$ follows immediately from the proof of proposition
\ref{**density}. By symmetry, $\rho_{p*}^{(n)}(B_n)$ is isomorphic to
$\rho_{*p}^{(n)}(B_n)$, thus this is a dense subgroup of $SU(f_n)$
modulo phase as well.
\end{proof}

\section{Fibonacci and Path Model Representations}
\label{rep_relations}

For any braid group $B_n$, and any root of unity $e^{i 2 \pi/k}$, the
path model representation is a homomorphism from $B_n$ to to a set of
linear operators. The vector space that these linear operators act is
the space of formal linear combinations of $n$ step paths on the rungs
of a ladder of height $k-1$ that start on the bottom rung. As an
example, all the paths for $n=4$, $k=5$ are shown in below.
\[
\includegraphics[width=5in]{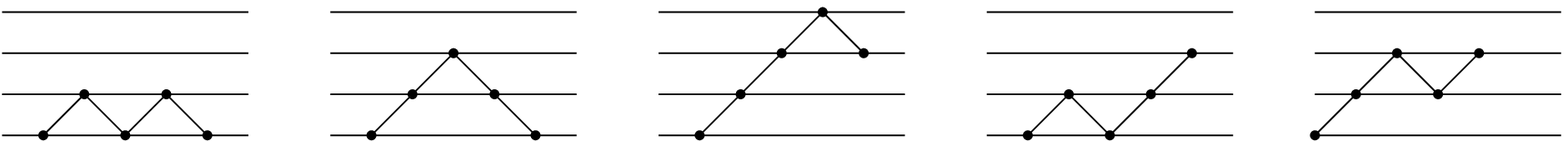}
\]
Thus, the $n=4, k=5$ path model representation is on a five
dimensional vector space. For $k=5$ we can make a bijective
correspondence between the allowed paths of $n$ steps and the set of
strings of $p$ and $*$ symbols of length $n+1$ which start with $*$
and have no to $*$ symbols in a row. To do this, simply label the
rungs from top to bottom as $*$, $p$, $p$, $*$, and directly read off
the symbol string for each path as shown below.
\[
\includegraphics[width=5in]{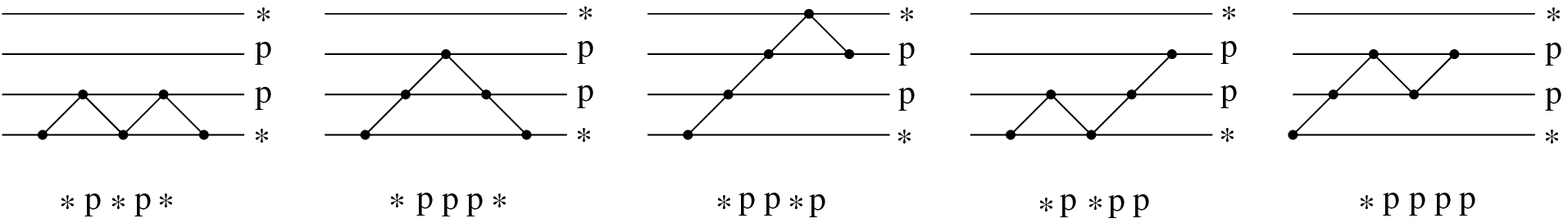}
\]
In \cite{Aharonov1}, it is explained in detail for any given braid how
to calculate the corresponding linear transformation on paths. Using
the correspondence described above, one finds that the path model
representation for $k=5$ is equal to the $-1$ times Fibonacci
representation described in this paper. This sign difference is a
minor detail which arises only because \cite{Aharonov1} chooses a
different fourth root of $t$ for $A$ than we do. This sign difference
is automatically compensated for in the factor of $(-A)^{3
  \cdot \mathrm{writhe}}$, so that both methods yield the correct Jones
polynomial.

\section{Unitaries on Logarithmically Many Strands}
\label{app_efficiency}

In this appendix we'll prove the following proposition.
\begin{prop1}
Given any pair of elements $U_{*p} \in SU(f_{k+1})$ and $U_{**} \in
SU(f_k)$, and any real parameter $\epsilon$, one can in polynomial
time find a braid $b \in B_k$ with $\mathrm{poly}(n,\log(1/\epsilon))$
crossings whose Fibonacci representation satisfies
$\| \rho_{*p}(b) - U_{*p} \| \leq \epsilon$ and 
$\| \rho_{**}(b) - U_{**} \| \leq \epsilon$, provided that $k = O(\log
n)$. By symmetry, the same holds when considering $\rho_{p*}$ rather
than $\rho_{*p}$.
\end{prop1}
To do so, we'll use a recursive construction. Suppose that we already
know how to achieve proposition \ref{efficiency} on $n$ symbols, and
we wish to extend this to $n+1$ symbols. Using the construction for
$n$ symbols we can efficiently obtain a unitary of the form
\begin{equation}
\label{generic}
M_{n-1}(A,B) = \begin{array}{l}
  \includegraphics[width=2in]{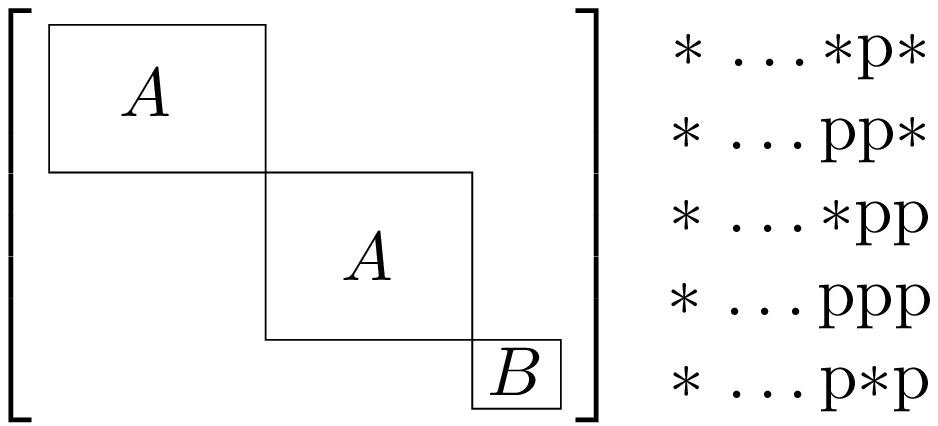} \end{array}
\end{equation}
where $A$ and $B$ are arbitrary unitaries of the appropriate
dimension. The elementary crossing $\sigma_n$ on the last two strands
has the representation\footnote{Here and throughout this appendix when
  we write a scalar $\alpha$ in a block of the matrix we really mean
  $\alpha I$ where $I$ is the identity operator of appropriate dimension.}
\[
M_n = \left[ \begin{array}{ccccc}
b &   &   &   &   \\
  & a &   &   &   \\
  &   & a &   &   \\
  &   &   & e & d \\
  &   &   & d & c
\end{array} \right]
\begin{array}{c}
\textrm{$*$ \ldots $*$p$*$} \\
\textrm{$*$ \ldots pp$*$} \\
\textrm{$*$ \ldots $*$pp} \\
\textrm{$*$ \ldots ppp} \\
\textrm{$*$ \ldots p$*$p}
\end{array}.
\]
As a special case of equation \ref{generic}, we can obtain
\[
M_{\mathrm{diag}}(\alpha) = \left[ \begin{array}{ccccc}
e^{i \alpha/2} &   &   &   &   \\
  & e^{i \alpha/2} &   &   &   \\
  &   & e^{i \alpha/2} &   &   \\
  &   &   & e^{i \alpha/2} &   \\
  &   &   &   & e^{-i \alpha/2}
\end{array} \right]
\begin{array}{c}
\textrm{$*$ \ldots $*$p$*$} \\
\textrm{$*$ \ldots pp$*$} \\
\textrm{$*$ \ldots $*$pp} \\
\textrm{$*$ \ldots ppp} \\
\textrm{$*$ \ldots p$*$p}
\end{array}.
\]
Where $0 \leq \alpha < 2 \pi$. We'll now show the following.
\begin{lemma}
\label{twodim}
For any element 
\[
\left[ \begin{array}{cc}
V_{11} & V_{12} \\
V_{21} & V_{22}
\end{array} \right] \in SU(2),
\]
one can find some product $P$ of $O(1)$ $M_{\mathrm{diag}}$ matrices and
$M_n$ matrices such that for some phases $\phi_1$ and $\phi_2$,   
\[
P = \left[ \begin{array}{ccccc}
\phi_1 &        &        &        &   \\
       & \phi_2 &        &        &   \\
       &        & \phi_2 &        &   \\
       &        &        & V_{11} & V_{12} \\
       &        &        & V_{21} & V_{22}
\end{array} \right]
\begin{array}{c}
\textrm{$*$ \ldots $*$p$*$} \\
\textrm{$*$ \ldots pp$*$} \\
\textrm{$*$ \ldots $*$pp} \\
\textrm{$*$ \ldots ppp} \\
\textrm{$*$ \ldots p$*$p}
\end{array}.
\]
\end{lemma}
\begin{proof}
Let $B_{\mathrm{diag}}(\alpha)$ and $B_n$ be the following $2 \times
2$ matrices 
\[
B_{\mathrm{diag}}(\alpha) = \left[ \begin{array}{cc}
e^{i \alpha/2} & 0 \\
0 & e^{- i \alpha/2}
\end{array} \right]
\quad \mathrm{and} \quad
B_n = \left[ \begin{array}{cc}
e & d \\
d & c
\end{array}
\right]
\]
We wish to show that we can approximate an arbitrary element of $SU(2)$ as
a product of $O(1)$ $B_{\mathrm{diag}}$ and $B_n$ matrices. To do this, we will
use the well known homomorphism $\phi: SU(2) \to SO(3)$ whose kernel
is $\{ \pm \id \}$ (see appendix \ref{density}). To obtain an arbitrary
element $V$ of $SU(2)$ modulo phase it suffices to show that
the we can use $\phi(B_n)$ and $\phi(B_{\mathrm{diag}}(\alpha))$ to
obtain an arbitrary $SO(3)$ rotation. In appendix \ref{density} we
showed that
\[
\left[ \begin{array}{cc}
a & 0 \\
0 & b
\end{array} \right]
\quad \mathrm{and} \quad
\left[ \begin{array}{cc}
e & d \\
d & c
\end{array}
\right]
\]
correspond to two rotations of $7 \pi/5$ about axes which are
separated by an angle of $\theta_{12} \simeq 1.8091137886\ldots$
By the definition of $\phi$, $\phi(B_{\mathrm{diag}}(\alpha))$ is a rotation
by angle $\alpha$ about the same axis that 
$\phi \left( \left[ \begin{array}{cc} a & 0 \\ 0 & b \end{array}
  \right] \right)$ rotates about. $\phi(B_n^5)$ is a $\pi$
rotation. Hence, $R(\alpha) \equiv \phi(B_n^5 
B_{\mathrm{diag}}(\alpha) B_n^5)$ is a rotation by angle $\alpha$
about an axis which is separated by angle\footnote{We subtract $\pi$
  because the angle between axes of rotation is only defined modulo
  $\pi$. Our convention is that these angles are in $[0,\pi)$.} $2
\theta_{12}-\pi$ from the 
axis that $\phi(B_{\mathrm{diag}}(\alpha))$ rotates about. $Q \equiv
R(\pi) \phi(B_{\mathrm{diag}}(\alpha)) R(\pi)$ is a rotation by angle
$\alpha$ about some axis whose angle of separation from the axis that
$\phi(B_{\mathrm{diag}}(\alpha))$ rotates about is $2 (2
\theta_{12}-\pi) \simeq 0.9532$. Similarly, by geometric
visualization, $\phi(B_{\mathrm{diag}}(\alpha')) Q
\phi(B_{\mathrm{diag}}(-\alpha'))$ is a rotation by $\alpha$
about an axis whose angle of separation from the axis that $Q$ rotates
about is anywhere from $0$ to $2 \times 0.9532$ depending on the value
of $\alpha'$. Since $2 \times 0.9532 > \pi/2$, there exists 
some choice of $\alpha'$ such that this angle of separation is
$\pi/2$. Thus, using the rotations we have constructed we can perform
Euler rotations to obtain an arbitrary rotation. 
\end{proof}

As a special case of lemma \ref{twodim}, we can obtain, up to global
phase,
\[
M_{\mathrm{swap}} = \left[ \begin{array}{ccccc}
\phi_1 &        &        &   &   \\
       & \phi_2 &        &   &   \\
       &        & \phi_2 &   &   \\
       &        &        & 0 & 1 \\
       &        &        & 1 & 0
\end{array} \right]
\begin{array}{c}
\textrm{$*$ \ldots $*$p$*$} \\
\textrm{$*$ \ldots pp$*$} \\
\textrm{$*$ \ldots $*$pp} \\
\textrm{$*$ \ldots ppp} \\
\textrm{$*$ \ldots p$*$p}
\end{array}.
\]
Similarly, we can produce $M_{\mathrm{swap}}^{-1}$. Using
$M_{\mathrm{swap}}$, $M_{\mathrm{swap}}^{-1}$, and equation
\ref{generic} we can produce the matrix
\[
M_C = \begin{array}{l} \includegraphics[width=2in]{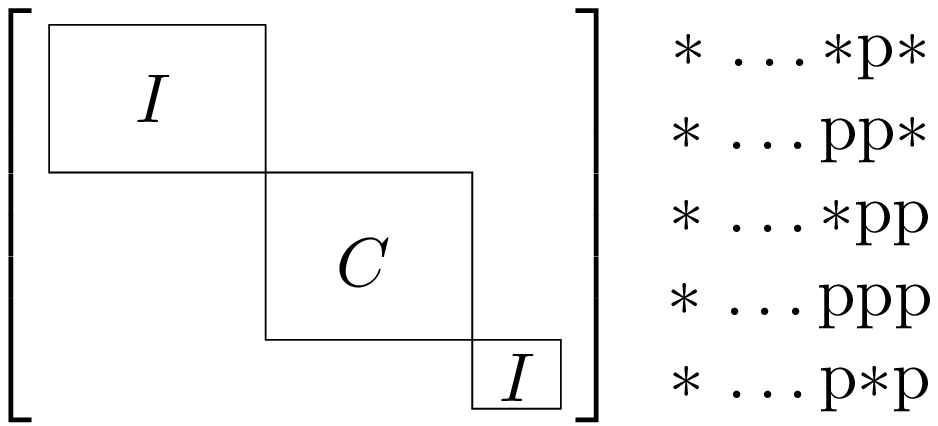} \end{array}
\]
for any unitary $C$. We do it as follows. Since $C$ is a normal
operator, it can be unitarily diagonalized. That is, there exists some
unitary $U$ such that $U C U^{-1} = D$ for some diagonal unitary
$D$. Next, note that in equation \ref{generic} the dimension of $B$ is
more than half that of $A$. Let $d = \mathrm{dim}(A) -
\mathrm{dim}(B)$, and let $I_d$ be the identity operator of dimension
$d$. We can easily construct two diagonal unitaries $D_1$ and $D_2$ of
dimension $\mathrm{dim}(B)$ such that  $(D_1 \oplus I_d)(I_d \oplus
D_2) = D$. As special cases of equation \ref{generic} we can obtain
\[
M_{D_1} = \left[ \begin{array}{ccccc}
1 &   &   &   &   \\
  & 1 &   &   &   \\
  &   & 1 &   &   \\
  &   &   & 1 &  \\
  &   &   &   & D_1
\end{array} \right]
\begin{array}{c}
\textrm{$*$ \ldots $*$p$*$} \\
\textrm{$*$ \ldots pp$*$} \\
\textrm{$*$ \ldots $*$pp} \\
\textrm{$*$ \ldots ppp} \\
\textrm{$*$ \ldots p$*$p}
\end{array}
\]
and
\[
M_{D_2} = \left[ \begin{array}{ccccc}
1 &   &   &   &   \\
  & 1 &   &   &   \\
  &   & 1 &   &   \\
  &   &   & 1 &  \\
  &   &   &   & D_2
\end{array} \right]
\begin{array}{c}
\textrm{$*$ \ldots $*$p$*$} \\
\textrm{$*$ \ldots pp$*$} \\
\textrm{$*$ \ldots $*$pp} \\
\textrm{$*$ \ldots ppp} \\
\textrm{$*$ \ldots p$*$p}
\end{array}
\]
and
\[
M_P = \begin{array}{l} \includegraphics[width=2in]{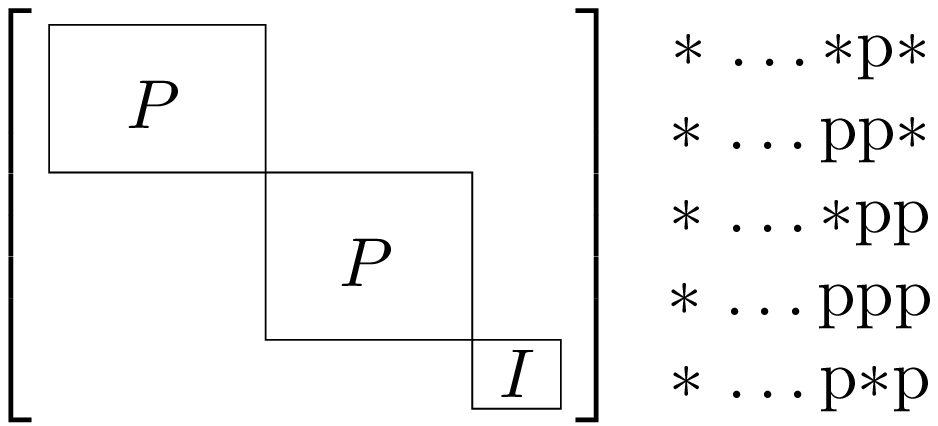} \end{array}
\]
where $P$ is a permutation matrix that shifts the lowest $\dim(B)$
basis states from the bottom of the block to the top of the block.
Thus we obtain
\[
M_2 \equiv M_{\mathrm{swap}} M_{D_2} M_{\mathrm{swap}}^{-1} =
\begin{array}{l}
\includegraphics[width=2in]{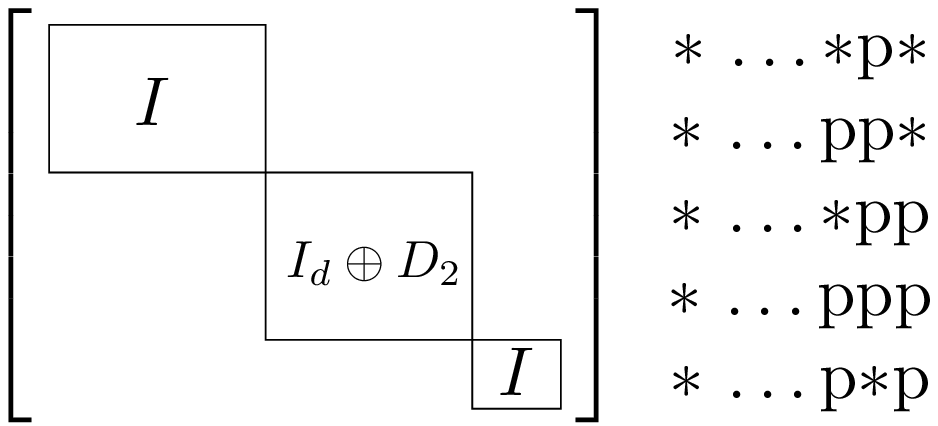}
\end{array}
\]
and
\[
M_1 \equiv M_P M_{\mathrm{swap}} M_{D_1}
M_{\mathrm{swap}}^{-1} M_P^{-1} =
\begin{array}{l}
\includegraphics[width=2in]{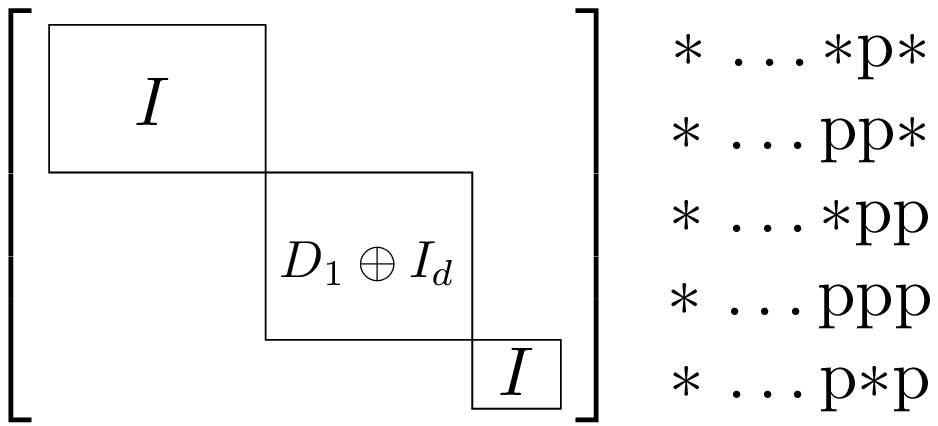}
\end{array}.
\]
Thus
\[
M_1 M_2 = \begin{array}{l}
 \includegraphics[width=2in]{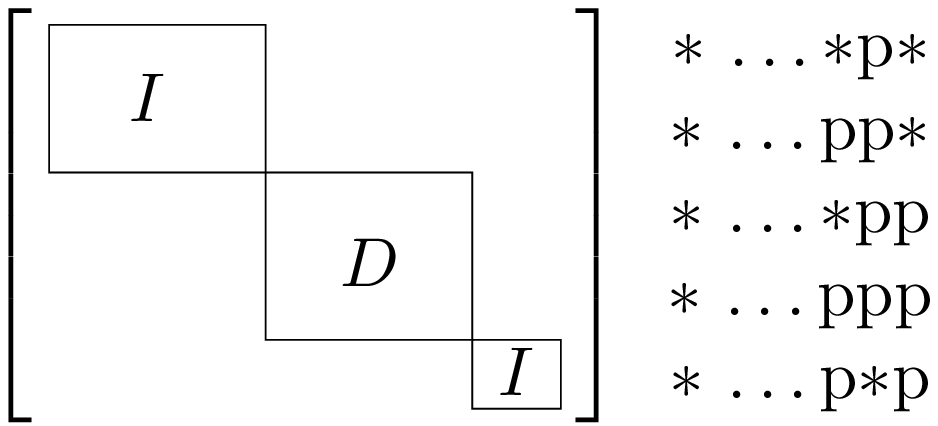}
 \end{array}.
\]
As a special case of equation \ref{generic} we can obtain
\[
M_U = \begin{array}{l}
 \includegraphics[width=2in]{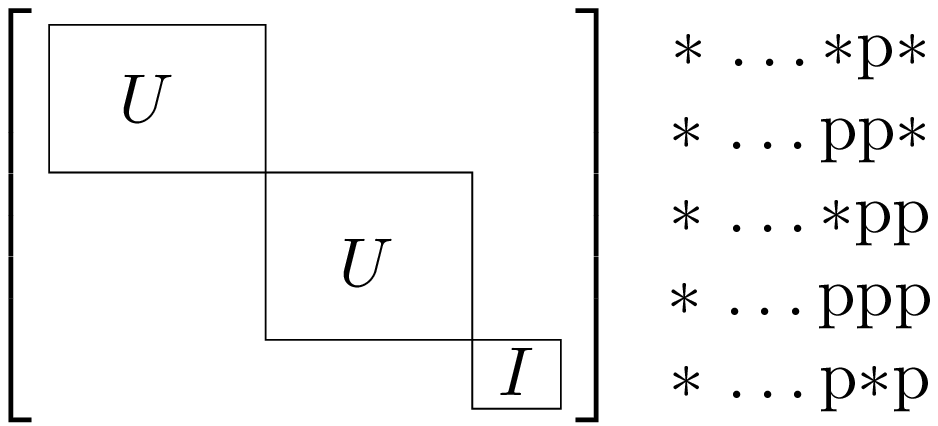}
 \end{array}.
\]
Thus we obtain $M_C$ by the construction $M_C = M_U M_1 M_2
M_U^{-1}$. By multiplying together $M_C$ and $M_{n-1}$ we can control
the three blocks independently. For arbitrary unitaries $A,B,C$ of
appropriate dimension we can obtain
\begin{equation}
\label{indep}
M_{ACB} = \begin{array}{l}
 \includegraphics[width=2in]{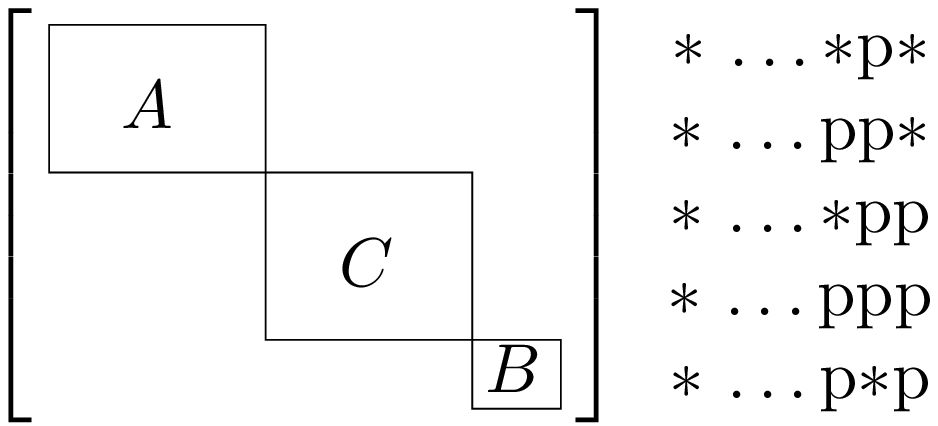}
 \end{array}.
\end{equation}
As a special case of equation \ref{indep} we can obtain
\[
M_{\mathrm{unphase}} =
\left[ \begin{array}{ccccc}
\phi_1^* &          &          &   &   \\
         & \phi_2^* &          &   &   \\
         &          & \phi_2^* &   &   \\
         &          &          & 1 &   \\
         &          &          &   & 1
\end{array} \right]
\begin{array}{c}
\textrm{$*$ \ldots $*$p$*$} \\
\textrm{$*$ \ldots pp$*$} \\
\textrm{$*$ \ldots $*$pp} \\
\textrm{$*$ \ldots ppp} \\
\textrm{$*$ \ldots p$*$p}
\end{array}.
\]
Thus, we obtain a clean swap
\begin{equation}
\label{clean}
M_{\mathrm{clean}} = M_{\mathrm{unphase}} M_{\mathrm{swap}} =
\left[ \begin{array}{ccccc}
1 &   &   &   &   \\
  & 1 &   &   &   \\
  &   & 1 &   &   \\
  &   &   & 0 & 1 \\
  &   &   & 1 & 0
\end{array} \right]
\begin{array}{c}
\textrm{$*$ \ldots $*$p$*$} \\
\textrm{$*$ \ldots pp$*$} \\
\textrm{$*$ \ldots $*$pp} \\
\textrm{$*$ \ldots ppp} \\
\textrm{$*$ \ldots p$*$p}
\end{array}.
\end{equation}
We'll now use $M_{\mathrm{clean}}$ and $M_{ACB}$ as our building
blocks to create the maximally general unitary
\begin{equation}
\label{maxgen}
M_{\mathrm{gen}}(V,W) = \begin{array}{l}
 \includegraphics[width=2in]{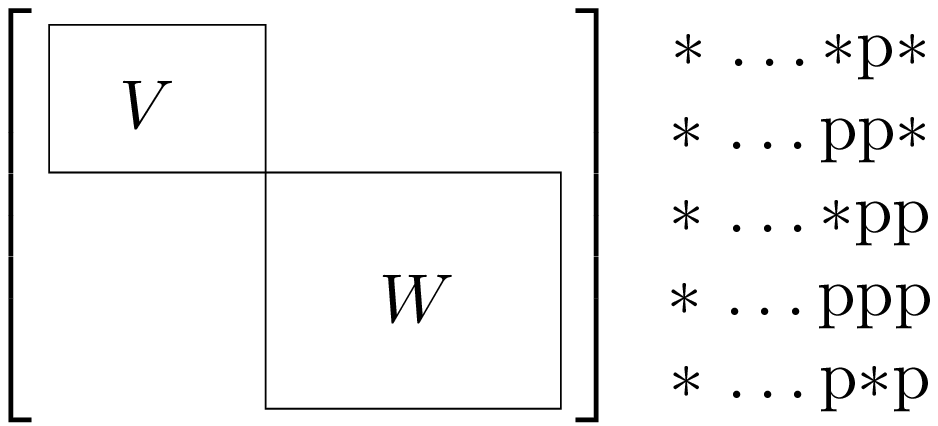}
 \end{array}.
\end{equation}

For $n+1$ symbols, the $* \ldots *pp$
subspace has dimension $f_{n-3}$, and the $* \ldots  p*p$ and $*
\ldots ppp$ subspaces each have dimension $f_{n-2}$. Thus, in equation
\ref{indep}, the block $C$ has dimension $f_{n-2} + f_{n-3} =
f_{n-1}$, and the block $B$ has dimension $f_{n-2}$. To construct
$M_{\mathrm{gen}}(V,W)$ we will choose a subset of the $f_n$ basis
states acted upon by the $B$ and $C$ blocks and permute them
into the $C$ block. Then using $M_{ACB}$, we'll perform an arbitrary unitary
on these basis states. At each such step we can act upon a subspace
whose dimension is a constant fraction of the dimension of the entire
$f_n$ dimensional space on which we wish to apply an arbitrary
unitary. Furthermore, this constant fraction is more than
half. Specifically, $f_n/f_{n-1} \simeq 1/\phi \simeq 0.62$ for large
$n$. We'll show that an arbitrary unitary can be built up as a product
of a constant number of unitaries each of which act only on half the
basis states. Thus our ability to act on approximately $62\%$ of the
basis states at each step is more than sufficient. 

Before proving this, we'll show how to permute an arbitrary set of
basis states into the $C$ block of $M_{ACB}$. Just use
$M_{\mathrm{clean}}$ to  swap the $B$ block into the $* \ldots ppp$
subspace of the $C$ block. Then, as a special case of equation
\ref{indep}, choose $A$ and $B$ to be the identity, and $C$ to be a
permutation which swaps some states between the $* \ldots *pp$ and $*
\ldots ppp$ subspaces of the $C$ block. The states which we swap up
from the $* \ldots ppp$ subspace are the ones from $B$ which we wish
to move into $C$. The ones which we swap down from the $* \ldots *pp$
subspace are the ones from $C$ which we wish to move into $B$. This
process allows us to swap a maximum of $f_{n-3}$ states between the
$B$ block and the $C$ block. Since $f_{n-3}$ is more than half the
dimension of the $B$ block, it follows that any desired permutation of
states between the $B$ and $C$ blocks can be achieved using two
repetitions of this process.

We'll now show the following.
\begin{lemma}
Let $m$ by divisible by 4. Any $m \times m$ unitary can be obtained as
a product of seven unitaries, each of which act only on the space
spanned by $m/2$ of the basis states, and leave the rest of the basis
states undisturbed.
\end{lemma}
It will be obvious from the proof that even if the dimension of the
matrix is not divisible by four, and the fraction of the basis states on
which the individual unitaries act is not exactly $1/2$ it will still
be possible to obtain an arbitrary unitary using a constant number of
steps independent of $m$. Therefore, we will not explicitly work out this
straightforward generalization.
\\
\begin{proof}
In \cite{Nielsen} it is shown that for any unitary $U$, one can always
find a series of unitaries $L_n, \ldots, L_1$ which each act on only
two basis states such that $L_n \ldots L_1 U$ is the identity. Thus
$L_n \ldots L_1 = U^{-1}$. It follows that any unitary can be
obtained as a product of such two level unitaries. The individual
matrices $L_1, \ldots, L_n$ each perform a (unitary) row
operation on $U$. The sequence $L_n \ldots L_1$ reduces $U$ to
the identity by a method very similar to Gaussian elimination. We will
use a very similar construction to prove the present lemma. The
essential difference is that we must perform the two level unitaries
in groups. That is, we choose some set of $m/2$ basis states, perform
a series of two level unitaries on them, then choose another set of
$m/2$ basis states, perform a series of two level unitaries on them,
and so on. After a finite number of such steps (it turns out that
seven will suffice) we will reduce $U$ to the identity.

Our two-level unitaries will all be of the same type. We'll fix our
attention on two entries in $U$ taken from a particular column:
$U_{ik}$ and $U_{jk}$. We wish to perform a unitary row operation,
\emph{i.e.} left multiply by a two level unitary, to set
$U_{jk}=0$. If $U_{ik}$ and $U_{jk}$ are not both zero, then the
two-level unitary which acts on the rows $i$ and $j$ according to
\begin{equation}
\label{row_op}
\frac{1}{\sqrt{|U_{ik}|^2 + |U_{jk}|^2}} 
\left[ \begin{array}{cc}
U_{ik}^* & U_{jk}^* \\
U_{jk} & - U_{ik}
\end{array} \right]
\end{equation}
will achieve this. If $U_{ik}$ and $U_{jk}$ are both zero there is
nothing to be done.

We can now use this two level operation within groups of basis states
to eliminate matrix elements of $U$ one by one. As in Gaussian
elimination, the key is that once you've obtained some zero matrix
elements, your subsequent row operations must be chosen so that they
do not make these nonzero again, undoing your previous work.

As the first step, we'll act on the top $m/2$ rows in order to reduce
the upper-left quadrant of $U$ to upper triangular form. We can do
this as follows. Consider the first and second entries in the first
column. Using the operation \ref{row_op} we can make the second entry
zero. Next consider the first and third entries in the first
column. By operation \ref{row_op} we can similarly make the third
entry zero. Repeating this procedure, we get all of the entries in the
top half of the first column to be zero other than the top
entry. Next, we perform the same procedure on the second column except
leaving out the top row. These row operations will not alter the first
column since the rows being acted upon all have zero in the first
column. We can then repeat this procedure for each column in the left
half of $U$ until the upper-left block is upper triangular. 

We'll now think of U in terms of 16 blocks of size
$(m/4) \times (m/4)$. In the second step we'll eliminate the matrix
elements in the third block of the first column. The second step is
shown schematically as
\[
\begin{array}{l} \includegraphics[width=2.0in]{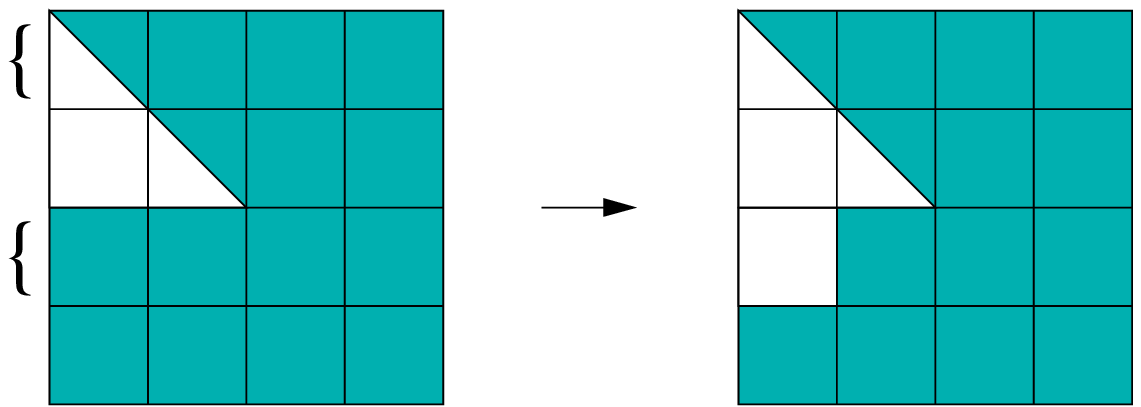} \end{array}
\]
The curly braces indicate the rows to be acted upon, and the unshaded
areas represent zero matrix elements. This step can be performed very
similarly to the first step. The nonzero matrix elements in the bottom
part of the first column can be eliminated one by one by interacting
with the first row. The nonzero matrix elements in the bottom part of
the second column can then be eliminated one by one by interacting
with the second row. The first column will be undisturbed by this
because the rows being acted upon in this step have zero matrix
elements in the first column. Similarly acting on the remaining
columns yields the desired result.

The next step, as shown below, is nearly identical and can be done the
same way.
\begin{equation}
\label{step3_eq}
\begin{array}{l} \includegraphics[width=2.0in]{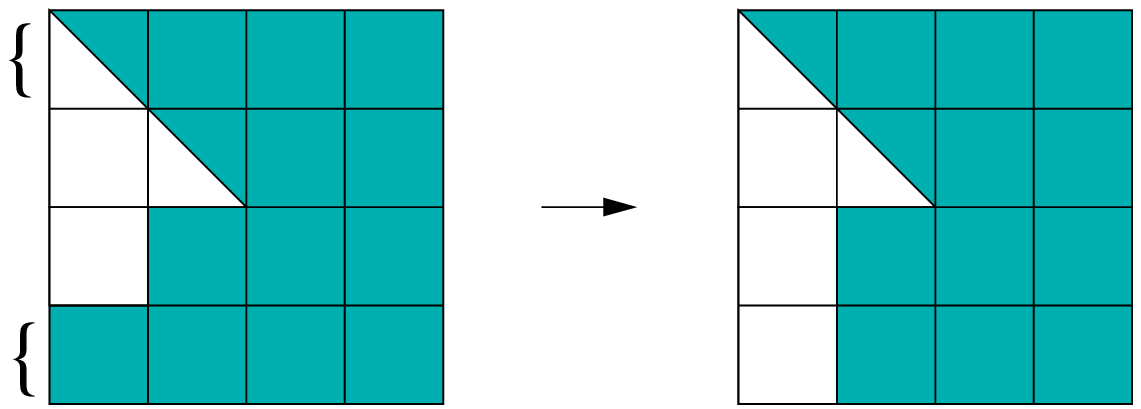} \end{array}
\end{equation}
The matrix on the right hand side of \ref{step3_eq} is unitary. It
follows that it must be of the form
\[
\begin{array}{l} \includegraphics[width=0.7in]{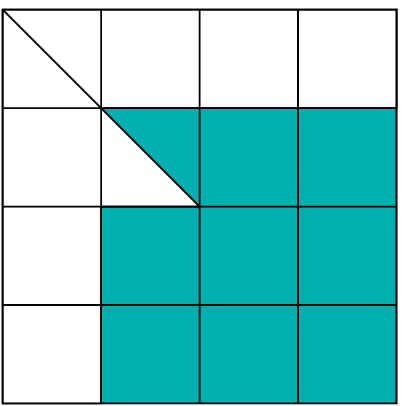} \end{array}
\]
where the upper-leftmost block is a diagonal unitary. We can next
apply the same sorts of steps to the lower $3 \times 3$ blocks, as
illustrated below.
\[
\begin{array}{l} \includegraphics[width=2.9in]{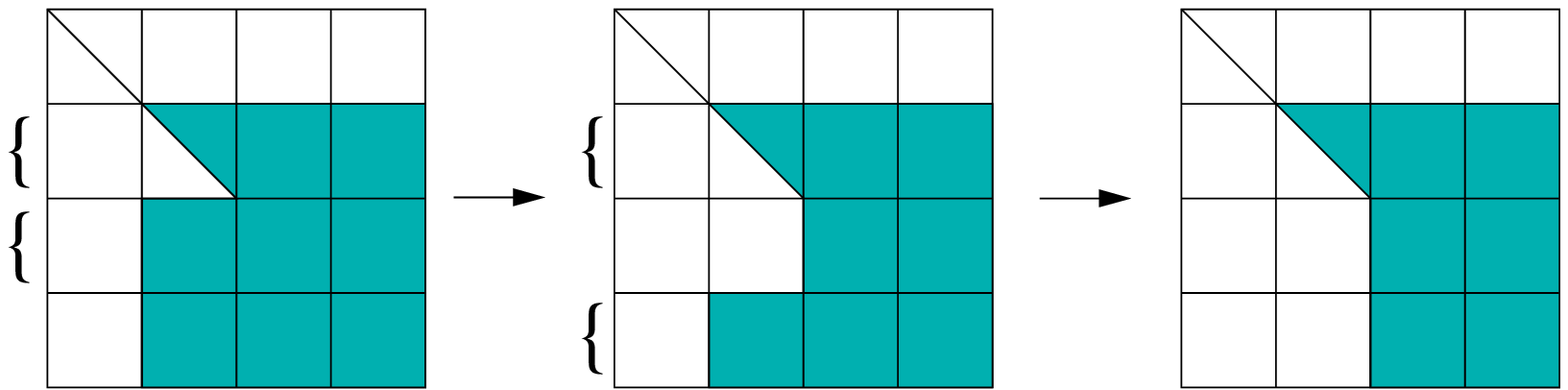} \end{array}
\]
By unitarity the resulting matrix is actually of the form
\[
\begin{array}{l} \includegraphics[width=0.7in]{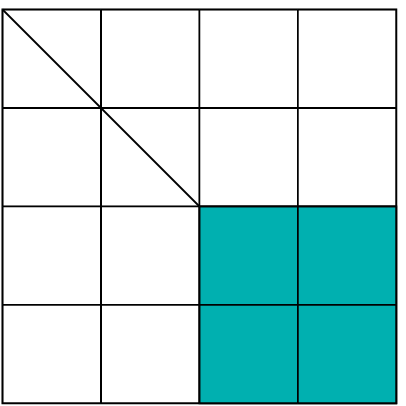} \end{array}
\]
where the lower-right quadrant is an $(m/2) \times (m/2)$ unitary matrix,
and the upper-left quadrant is an $(m/2) \times (m/2)$ diagonal unitary
matrix. We can now apply the inverse of the upper-left quadrant to the
top $m/2$ rows and then apply the inverse of the lower-right quadrant
to the bottom $m/2$ rows. This results in the identity matrix, and we
are done. In total we have used seven steps.
\end{proof}

Examining the preceding construction, we can see the recursive step
uses a constant number of the $M_{n-1}$ operators from the next lower
level of recursion, plus a constant number of $M_n$ operators. Thus,
the number of crossings in the braid grows only exponentially in the
recursion depth. Since each recursion adds one more symbol, we see
that to construct $M_{\mathrm{gen}}(V,W)$ on logarithmically many
symbols requires only polynomially many crossings in the corresponding
braid.

The main remaining task is to work out the base case on which the
recursion rests. Since the base case is for a fixed set of generators
on a fixed number of symbols, we can simply use the Solovay-Kitaev
theorem\cite{Kitaev}.
\begin{theorem}[Solovay-Kitaev]
\label{Solovay-Kitaev}
Suppose matrices $U_1, \ldots, U_r$ generate a dense subgroup in
$SU(d)$. Then, given a desired unitary $U \in SU(d)$, and a precision
parameter $\delta > 0$, there is an algorithm to find a product $V$ of
$U_1, \ldots, U_r$ and their inverses such that $\| V - U \| \leq
\delta$. The length of the product and the runtime of the algorithm
are both polynomial in $\log(1/\delta)$.
\end{theorem}

Because the total complexity of the process is polynomial, it
is only necessary to implement the base case to polynomially small
$\delta$ in order for the final unitary $M_{\mathrm{gen}}(V,W)$ to
have polynomial precision. This follows from simple error propagation.
An analogous statement about the precision of gates needed in
quantum circuits is worked out in \cite{Nielsen}. This completes the
proof of proposition \ref{efficiency}.

\section{Zeckendorf Representation}
\label{bijective}

Following \cite{Kauffman}, to construct the Fibonacci representation
of the braid group, we use strings of p and $*$ symbols such that no
two $*$ symbols are adjacent. There exists a bijection $z$ between
such strings and the integers, known as the Zeckendorf
representation. Let $P_n$ be the set of all such strings of length
$n$. To construct the map $z:P_n \to \{0,1,\ldots,f_{n+2} \}$ 
we think of $*$ as one and $p$ as zero. Then, for a given 
string $s = s_n s_{n-1} \ldots s_1$ we associate the integer
\begin{equation}
\label{recap}
z(s) = \sum_{i=1}^n s_i f_{i+1},
\end{equation}
where $f_i$ is the $i\th$ Fibonacci number: $f_1 = 1, f_2=1, f_3 = 2$,
and so on. In this appendix we'll show the following.
\begin{proposition}
\label{bij}
For any $n$, the map $z: P_n \to \{0, \ldots, f_{n+2} \}$ defined by
  $z(s) = \sum_{i=1}^n s_i f_{i+1}$ is bijective.
\end{proposition}

\begin{proof}
We'll inductively show that the following two statements
are true for every $n \geq 2$. 
\\ \\
$\mathbf{A_n:}\quad$
\emph{$z$ maps strings of length $n$ starting with p bijectively to
  $\{0,\ldots, f_{n+1}-1 \}$.}
\\ \\
$\mathbf{B_n:}\quad$
\emph{$z$ maps strings of length $n$ starting with $*$ bijectively to
  $\{f_{n+1}, \ldots, f_{n+2}-1 \}$.}
\\ \\
Together, $A_n$ and $B_n$ imply that $z$ maps $P_n$ bijectively to
$\{0, \ldots, f_{n+2}-1 \}$. As a base case, we can look at $n=2$. 
\begin{eqnarray*}
pp & \leftrightarrow & 0 \\
p* & \leftrightarrow & 1 \\
*p & \leftrightarrow & 2
\end{eqnarray*}
Thus $A_2$ and $B_2$ are true. Now for the induction. Let $s_{n-1} \in
P_{n-1}$. By equation 
\ref{recap},
\[
z(p s_{n-1}) = z(s_{n-1}).
\]
Since $s_{n-1}$ follows a p symbol, it can be any element of
$P_{n-1}$. By induction, $z$ is bijective on $P_{n-1}$, thus $A_n$ is
true. Similarly, by equation 
\ref{recap}
\[
z(* s_{n-1}) = f_{n+1} + z(s_{n-1}).
\]
Since $s_{n-1}$ here follows a $*$, its allowed values are exactly
those strings which start with p. By induction, $A_{n-1}$ tells us
that $z$ maps these bijectively to $\{0, \ldots, f_n-1 \}$. Since
$f_{n+1} + f_n = f_{n+2}$, this implies $B_n$ is true. Together, $A_n$
and $B_n$ for all $n \geq 2$, along with the trivial $n=1$ case,
imply proposition \ref{bij}.
\end{proof}

\bibliography{jones}

\end{document}